\RequirePackage{snapshot}       

\PassOptionsToPackage{inline}{enumitem}
\documentclass{lmcs}
\pdfoutput=1

\usepackage{lastpage}
\lmcsdoi{15}{2}{10}
\lmcsheading{}{\pageref{LastPage}}{}{}%
{Nov.~06,~2017}{May~08,~2019}{}

\ACMCCS{[\textbf{Theory of computation}]: Denotational semantics; Concurrency}

\usepackage[utf8]{inputenc}
\usepackage{mathtools}
\usepackage{stmaryrd}
\usepackage[all,cmtip]{xy}

\def\eg{{e.g.\@}}

\def\ie{{i.e.\@}}

\setlist[description]{font=\normalfont\scshape}

\DeclareMathOperator{\dom}{dom}

\newcommand{\VExp}{\mathsf{VExp}}
\newcommand{\BExp}{\mathsf{BExp}}
\newcommand{\Cmd}{\mathsf{Cmd}}

\newcommand{\Bool}{\mathrm{Bool}}

\newcommand{\cskip}{\mathbf{skip}}
\newcommand{\cfence}{\mathbf{fence}}
\newcommand{\ifb}[3]{\mathbf{if} \ #1\ \mathbf{then}\ #2\ \mathbf{else}\ #3}

\newcommand{\whileb}[2]{\mathbf{while}\ #1\ \mathbf{do}\ #2}
\DeclareMathOperator{\parop}{{\|}}

\newcommand{\Apo}{{\mathcal{A}_{PO}}}
\newcommand{\Atso}{{\mathcal{A}_{TSO}}}
\newcommand{\Aw}{{\mathcal{A}_w}}
\newcommand{\Ar}{{\mathcal{A}_r}}
\newcommand{\Ab}{{\mathcal{A}_b}}

\newcommand{\calP}{\mathcal{P}}
\newcommand{\PB}{\mathcal{B}}
\newcommand{\poms}{\mathrm{Pom}}
\newcommand{\Ppo}{{\mathcal{P}_{PO}}}

\newcommand{\lc}[2]{#1{\downarrow}_{#2}}

\renewcommand{\nc}{\mathrel{\parop}}
\newcommand{\comp}{\mathrel{{\leq}{\geq}}}

\newcommand{\lst}[1]{{[\,#1\,]}}
\newcommand{\emptylist}{{\lst{}}}
\newcommand{\lists}[1]{#1~\mathsf{list}}

\newcommand{\scons}[2]{#1 \mathrel{\Uparrow} #2}
\newcommand{\state}[1]{\left[#1\right]}
\newcommand{\supdm}[1]{\state{#1}}
\newcommand{\supd}[2]{\supdm{#1\mid#2}}

\newcommand{\Loc}{\mathrm{Loc}}
\newcommand{\BLoc}{\mathrm{BLoc}}
\newcommand{\Locs}{\mathrm{Locs}}
\newcommand{\pfin}{\rightharpoonup_{fin}}
\newcommand{\N}{\mathbb{N}}
\newcommand{\lsState}[1]{\beta_{#1}}

\newcommand{\emptystate}{{\state{\,}}}

\usepackage{amssymb}
\let\oldrestriction\restriction%
\renewcommand{\restriction}{{\oldrestriction}}

\DeclareMathOperator{\Lin}{Lin}
\newcommand{\sembr}[1]{\llbracket#1\rrbracket}
\newcommand{\footp}[1]{\sembr{#1}}
\newcommand{\zfootp}[1]{\footp{#1}^\zeta}
\newcommand{\pex}[1]{\sembr{#1}}
\newcommand{\PT}{\mathcal{P}}
\newcommand{\PTSO}{\PT_{TSO}}

\DeclareMathOperator{\seqop}{\triangleleft}
\renewcommand{\seqop}{\mathrel{{\triangleleft}}}

\newcommand{\emptypset}{\mathbf{0}}

\newcommand{\Ls}{\mathsf{Ls}}

\DeclareMathOperator{\seqit}{\curvearrowright}
\DeclareMathOperator{\splt}{split}
\newcommand{\truev}{\mathit{true}}
\newcommand{\falsev}{\mathit{false}}

\DeclareMathOperator{\charop}{\chi}
\newcommand{\charac}[2]{\charop_{#1}(#2)} 
\DeclareMathOperator{\differop}{\Delta}
\newcommand{\differ}[2]{\differop_{#1}(#2)} 
\newcommand{\maxo}[1]{\max\{#1,0\}}

\newcommand{\E}{\mathcal{E}}
\newcommand{\DExec}[1]{\E(#1)}

\newcommand{\T}{\mathcal{T}}

\newcommand{\aand}{\land}

\newcommand{\rn}[1]{(\textsc{#1})} 
\newcommand{\defin}[1]{{\usefont{T1}{lmss}{b}{n}{#1}}} 

\mathtoolsset{centercolon}

\begin{document}

\title{A Denotational Semantics for SPARC TSO}
\author{Ryan Kavanagh}
\author{Stephen Brookes}     

\address{Computer Science Department\\Carnegie Mellon University\\Pittsburgh, PA, USA}
\email{\{rkavanagh, brookes\}@cs.cmu.edu}
\thanks{Funded in part by a Natural Sciences and Engineering Research Council of Canada (NSERC) Postgraduate Scholarship.}


\begin{abstract}
  The SPARC TSO weak memory model is defined axiomatically, with a non-compositional formulation that makes modular reasoning about programs difficult.
  Our denotational approach uses pomsets to provide a compositional semantics capturing exactly the behaviours permitted by SPARC TSO\@.
  It uses buffered states and an inductive definition of execution to assign an input-output meaning to pomsets.
  We show that our denotational account is sound and complete relative to the axiomatic account, that is, that it captures exactly the behaviours permitted by the axiomatic account.
  Our compositional approach facilitates the study of SPARC TSO and supports modular analysis of program behaviour.
\end{abstract}

\maketitle

\section{Introduction}%
\label{sec:intro}

A memory model specifies what values can be read by a thread from a given memory location during execution.
Traditional concurrency research has assumed \textit{sequential consistency}, wherein memory actions operate atomically on a global state, and a read is guaranteed to observe the value most recently written to that location \textit{globally}.
Consequently, ``the result of any execution is the same as if the operations of all the processors were executed in some sequential order''~\cite{Lamport1979}.
However, sequential consistency negatively impacts performance, and modern architectures often provide much weaker guarantees.
These weaker guarantees mean that classical concurrency algorithms, which often assume sequential consistency, can behave in unexpected ways.
Consider, for example, the Dekker algorithm on a system using the SPARC instruction set.
The Dekker algorithm seeks to ensure that at most one process enters a critical section at a time.
Executing the following instance of the Dekker algorithm on a sequentially consistent system from an initial state where memory locations $w$, $x$, $y$, $z$ are all zero will ensure that we end in a state where not both $z$ and $w$ are set to one:
\[ (x := 1; \ifb{y = 0}{z:=1}{\cskip}) \parop (y:=1; \ifb{x=0}{w:=1}{\cskip}). \]
However, the SPARC ISA provides the weaker SPARC TSO (total store ordering) memory model.
Under SPARC TSO, it is possible to start from the aforementioned initial state and end in a state where both $z$ and $w$ are set to one, thus violating mutual exclusion.

Weak memory models are often described in standards documents using natural language.
This informality makes it difficult to reason about how programs will behave on systems that use these memory models.
The SPARC Architecture Manual~\cite{SPARC8} gives an axiomatic description of TSO using partial orders of actions.
We present this description in Section~\ref{sec:axiom}; other axiomatic approaches are discussed in Section~\ref{sec:related-work}.
Despite their formality, it is hard to use axiomatic accounts to reason about the behaviour of programs.
This is because the axiomatic approach is non-compositional and precludes modular reasoning.
We address this problem by presenting a denotational semantics for SPARC TSO in Section~\ref{sec:denotational}.
Our denotational semantics assigns to each program a collection of pomsets.
Pomsets are generalizations of traces and were first used by Pratt~\cite{Pratt1986} to give denotational semantics for concurrency, and later by Brookes~\cite{Brookes2016MFPS}, with some modifications, to study weak memory models.
We illustrate our semantics by validating various \textit{litmus tests} and expected program equivalences.
To ensure our denotational semantics accurately captures the behaviour of SPARC TSO, we show in Section~\ref{sec:sands} that from every denotation of a program we can derive a collection of partial orders satisfying the axiomatic description of Section~\ref{sec:axiom} and, moreover, that we can derive every such partial order from the denotation.

\section{An Axiomatic Account}%
\label{sec:axiom}

The SPARC manual~\cite{SPARC8} gives an axiomatic description of TSO in terms of partial orders of actions.
Unfortunately, this description is incomplete because it fails to specify the fork and join behaviour of TSO\@.
In this section, we complete the SPARC manual's account of TSO with an account of forking and joining.
Before doing so, we give an informal description of TSO to help build intuition.

A system providing the TSO weak memory model can be thought of as a collection of processors, each with a write buffer.
Whenever a processor performs a write, it places it in its write buffer.
The buffer behaves as a queue, and writes migrate out of the buffers one at a time, and shared memory applies them according to a global total order.
Whenever a processor tries to read a location, it first checks its buffer for a write to that location.
If it finds one, it uses the value of the most recent such write; otherwise, it looks to shared memory for a value.
Because of buffering, it is possible for writes and reads to be observed out of order relative to the program order.

\subsection{Program Order Pomsets}

To make the above intuition precise, we must formalize the notion of a \textit{program order}, \ie, the ordering of read and write actions specified by a program.
We do so by means of partially-ordered multisets.

\begin{defi}
A (strict) \defin{partially-ordered multiset} or \defin{pomset} $(P, <, \Phi)$ over a label set $L$ consists of a strict poset $(P, <)$ of ``action occurrences'' and a function $\Phi : P \to L$ mapping each action occurrence to its label or ``action''.
\end{defi}

\noindent We frequently write just $P$ for $(P, <, \Phi)$, in which case we let $<_P$ and $\Phi_P$ denote its obvious components.
Denote by $\poms(L)$ the set of pomsets over $L$\@.
We remark that pomsets are a natural generalization of traces.
Indeed, all traces are simply pomsets where the underlying order is total.

We do not usually make the poset $P$ explicit, because the structure of the pomset is invariant under relabellings of the elements of $P$\@.
Consequently, we identify pomsets $(P, <, \Phi)$ and $(P', <', \Phi')$ if there exists an order isomorphism $\phi : P \to P'$ such that $\Phi = \Phi' \circ \phi$\@.
We usually denote the elements of the pomset using just their labels, but we sometimes need to specify their exact occurrence, in which case we write $l_p$, where $l = \Phi(p)$\@.

It is useful to draw a pomset $P$ as a labelled directed acyclic graph, where multiple vertices can have the same label and we have an edge $a \to b$ if $a <_P b$\@.
For clarity, we always omit edges obtained by transitivity of $<_P$\@.
For example, the following graph depicts the pomset where $P = \{ 0, 1, 2, 3\}$, the order is given by $0 < 1$, $1 < 2$, $0 < 2$, and $0 < 3$, and $\Phi$ is given by $\Phi(0) = a$, $\Phi(1) = b$, $\Phi(2) = a$, and $\Phi(3) = c$:
\[\xymatrix{c & a \ar[l] \ar[r] & b \ar[r] &a}.\]

We assume a countably infinite set of locations $\Loc$, ranged over by metavariables $x, y, z, \dotsc$, and a set of values $V$, ranged over by $v$\@.
In our examples, we will take $V$ to be the set of integers.
We call $x:=v$ a \textit{global write action}, $x = v$ a \textit{read action}, and $\delta$ a \textit{skip action}.
Let $\Aw$ and $\Ar$ be the sets of global write actions and read actions, respectively.

\begin{defi}
  A \defin{program order} is a pomset $P$ over the set $\Apo = \Aw \cup \Ar \cup \{\delta\}$ of action labels that satisfies the \textit{(locally) finite height property}, that is, such that for all $b \in P$, the set $\{ a \in P \mid a <_P b \}$ is finite.
\end{defi}

\noindent Informally, the finite height property guarantees that all actions in $P$ can be executed after finitely many other actions, \ie, that the program order contains no unreachable actions.

Intuitively, the program order \[ \xymatrix{x:=2\ar[r] & y=1& y:=1\ar[r] & y=1} \] describes the parallel execution of writing 2 to $x$ before reading 1 from $y$, and writing 1 to $y$ before reading 1 from $y$, with no other ordering constraints.

\subsection{TSO Axioms}

We now turn our attention to our completed version of the axiomatic account given in the SPARC manual.
To do so, we first introduce the notion of state and the requisite notation.

A global state is a finite partial function from locations $\Loc$ to values $V$\@.
We let $\Sigma_{PO} = \Loc \pfin V$ be the set of global states, and use $\sigma$ to range over $\Sigma_{PO}$\@.

Given any set $S$ and partial order $<_S$ on it, every element $s \in S$ determines a set $\lc{s}{S} = \{ s' \in S \mid s' <_S s \} \cup \{s\}$ called its \textit{lower closure}.
Write $s \nc_S s'$ to denote that $s$ and $s'$ are not comparable under the reflexive closure $\leq_S$ of $<_S$, and write $s \comp_S s'$ to denote that they are comparable.

\begin{defi}%
  \label{def:ctso}
  Let $P$ be a program order and $<_T$ be a strict partial order on the elements of $P$\@.
  We say $<_T$ is \defin{TSO-consistent} with $P$ from (the initial state) $\sigma$ if it satisfies the following six axioms:
  \begin{description}[before={\renewcommand\makelabel[1]{##1}}]
  \item[\rn{O}] \textbf{Ordering:} $<_T$ totally orders the write actions $\Aw$ of $P$\@.
  \item[\rn{V}] \textbf{Value:} for all reads ${(x=v)}_r$ in $P$, either
    \begin{description}[before={\renewcommand\makelabel[1]{##1}}]
    \item[\rn{a}] there exists a write ${(x:=v')}_w$ maximal under $<_T$ amongst all writes~to~$x$~in $\lc{{(x=v)}_r}T$, all writes to $x$ in $\lc{{(x=v)}_r}P$ are in $\lc{{(x:=v')}_w}T$, and $v = v'$;~or
    \item[\rn{b}] there exists a write ${(x:=v')}_w$ maximal under $<_P$ amongst all writes to~$x$~in $\lc{{(x=v)}_r}P$, and both ${(x=v)}_r <_T {(x:=v')}_w$ and $v = v'$; or
    \item[\rn{c}] there are no writes to $x$ in $\lc{{(x=v)}_r}T$ or $\lc{{(x=v)}_r}P$, and $\sigma(x) = v$\@.
    \end{description}
  \item[\rn{L}] \textbf{LoadOp:} for all reads $r \in P$ and all actions $a \in P$, $r <_P a$ implies $r <_T a$\@.
  \item[\rn{S}] \textbf{StoreStore:} for all writes $w, w' \in P$, $w <_P w'$ implies $w <_T w'$\@.
  \item[\rn{F}] \textbf{Fork:} if $\alpha_1 <_P \alpha_2$, $\alpha_1 <_P \alpha_3$, and $\alpha_2 \nc_P \alpha_3$, then $\alpha_1 <_T \alpha_2$ and $\alpha_1 <_T \alpha_3$\@.
  \item[\rn{J}] \textbf{Join:} if $\alpha_1 <_P \alpha_3$, $\alpha_2 <_P \alpha_3$, and $\alpha_1 \nc_P \alpha_2$, then $\alpha_1 <_T \alpha_3$ and $\alpha_2 <_T \alpha_3$\@.
  \end{description}
  The fork axiom is easily understood by: if $\alpha_1 \leftarrow \alpha_1 \rightarrow \alpha_3$ in $P$, then $\alpha_2 \leftarrow \alpha_1 \rightarrow \alpha_3$ or $\alpha_1 \rightarrow \alpha_2 \rightarrow \alpha_3$ in $(P,{<_T})$; the join axiom is symmetric.
  We simply say $<_T$ is TSO-consistent with $P$ if there exists some initial state $\sigma$ from which they are TSO-consistent.
  It will be useful to identify $<_T$ and the pomset $T = (P, <_T, \Phi_P)$\@.
  \qed%
\end{defi}

\noindent Axioms \rn{O}, \rn{Va}, \rn{Vb}, \rn{L}, and \rn{S} are directly adapted from the formal specification given in Appendix~K.2 of~\cite{SPARC8}.
We introduce axiom \rn{Vc} to simplify our presentation of examples.
By requiring that programs first write to any locations from which they read, it can be omitted, and apart from examples, we will assume throughout that our TSO-consistent orders do not require \rn{Vc}.
Though the formal specification does not provide axioms \rn{F} and \rn{J}, they are consistent with the behaviour intended by Appendix~J.6 of~\cite{SPARC8}.
Intuitively, axiom \rn{Vb} requires that whenever a processor reads from a location, it must use the most recent write to that location in its write buffer (if it exists).
If there is no such write in its write buffer, but we have observed a global write to that location, then \rn{Va} requires that the most recent such write be the one read.
Our presentation differs slightly from the formal specification.
In particular, we do not consider instruction fetches or atomic load-store operations, and we do not consider flush actions, because they can be implemented as a derived action in our semantics by forking and immediately joining.
To be consistent with our pomset development, we also assume the order~to~be~strict.

As the following proposition's corollary shows, if $<_T$ is TSO-consistent for $P$, then there exists a (not necessarily unique) total order on $P$ that is TSO-consistent with $P$ and contains $<_T$\@.
As a result, we can view all orders that are TSO-consistent with $P$ as weakenings of total orders that are TSO-consistent with $P$\@.
The proposition follows by a straightforward check of the axioms, where \rn{V} is the only axiom that is not immediate.

\begin{prop}
  Let $<_T$ be TSO-consistent for $P$ and $a, b \in P$ be two actions such that $a \nc_T$ b.
  Let $<_{ab}$ be the transitive closure of ${<_T} \cup \{(a,b)\}$\@.
  If there exist maximum writes under $a$ and $b$ relative to $\leq_T$, call them $\mu_a$ and $\mu_b$, respectively.
  If either $\mu_a$ and $\mu_b$ exist and $\mu_a <_T \mu_b$, or $\mu_a$ does not exist, then $<_{ab}$ is TSO-consistent for $P$\@.
  \qed%
\end{prop}

\begin{cor}%
  \label{cor:tsotot}
  If $P$ is a finite program order pomset and $<_T$ a partial order TSO-consistent with $P$, then there exists a total order $\sqsubset_T$ TSO-consistent with $P$ such that ${<_T} \subseteq {\sqsubset_T}$\@.
  \qed%
\end{cor}

\noindent Despite Corollary~\ref{cor:tsotot}, one should be careful not to conflate the notion of linearisation with that of TSO-consistent total orders.
Consider, for example, the program order
\[ \xymatrix{x:=2\ar[r] & x= 2 & x:=3\ar[r] & x=3}.
\] The linearisation ${x:=2} < {x:=3} < {x = 3} < {x = 2}$ is not TSO-consistent with the program order because it violates \rn{Va}; the order ${x=2}< {x:=2}< {x = 3} < {x:=3}$ is not a linearisation of the program order but is TSO-consistent with it.

When we have a write followed by a read in the program order, but swapped in the linear order, as in this example, we can imagine the write having gotten stuck in the write buffer, and observing the read before the write.

\section{A Denotational Account}%
\label{sec:denotational}

So far we have dealt with program orders in the abstract.
To make the rest of our development more concrete, we restrict our attention to program orders for well-defined programs in the simple imperative language given below.
These program orders are defined~in~Section~\ref{subsec:popoms}.

Restricting our attention to program orders of these well-defined programs raises the question of \textit{compositionality}.
The key is to find a way to derive TSO-consistent orders for a sequential composition $c_1;c_2$ or parallel composition $c_1\parop c_2$ given TSO-consistent orders for $c_1$ and $c_2$\@.
This is infeasible with the axiomatic approach, which requires reasoning about whole programs and is inherently non-compositional.
In contrast, a denotational approach using pomsets is compositional: it allows us to derive the meaning of a program vis-à-vis a weak memory model from the meanings of its parts vis-à-vis the memory model.

Our denotational semantics has two components.
The first associates to each program a set of \textit{TSO pomsets}, which serves as the \textit{abstract meaning} or \textit{denotation }of the program.
This component is described in Section~\ref{subsec:tsopoms}.
The second associates to each pomset a set of \textit{executions}, which describe its input-output behaviours.
This is described in Section~\ref{sec:footprints}.

\subsection{A Simple Imperative Language}

We express our programs using a simple imperative language.
This formalism avoids the complexity of high-level languages, while still capturing the programs we are interested in.
In the syntax below, $e$ ranges over integer expressions, $b$ over boolean expressions, $c$ over commands, and $p$ over programs.
We distinguish between commands and programs, because although commands can be composed to form new commands, programs are assumed to be syntactically complete and executable.
This distinction will be useful later when we consider executions, where we will assume programs are executed from initial states with empty buffers, but impose no such constraint on executions for commands.
\begin{align*}
  v &::= \dotsc, -2, -1, 0, 1, 2, \dotsc\\
  e &::= v \mid x \mid e_1 + e_2 \mid e_1 * e_2 \mid \cdots\\
  b &::= \truev \mid \falsev \mid \neg b \mid e_1 = e_2 \mid e_1 < e_2 \mid b_1 \vee b_2 \mid b_1 \land b_2 \mid \cdots\\
  c &::= \cskip \mid x := e \mid c_1;c_2 \mid c_1 \parop c_2 \mid \ifb{b}{c_1}{c_2} \mid \whileb{b}{c}\\
  p &::= c
\end{align*}
Let $\VExp$ denote the set of integer expressions, $\BExp$ the set of boolean expressions, and $\Cmd$ the set of commands.

\subsection{PO Pomsets}%
\label{subsec:popoms}

Given a command $c$ in our language, we must now compile it down to its set $\Ppo(c)$ of program order pomsets.
We need operations for the sequential and parallel composition of pomsets over the same set of labels.
When defining compositions of pomsets, we assume without loss of generality that the underlying posets are disjoint.

\begin{defi}
  The \defin{sequential composition} $(P_0, {<_0}, \Phi_0);(P_1, {<_1}, \Phi_1)$ is $(P_0, {<_0}, \Phi_0)$ whenever $P_0$ is infinite, and otherwise it is $({P_0} \cup {P_1}, {<_0} \cup {<_1} \cup P_0 \times P_1, \Phi_0 \cup \Phi_1)$\@.
  The \defin{parallel composition} $(P_0, {<_0}, \Phi_0) \parop (P_1, {<_1}, \Phi_1)$ of pomsets is $(P_0 \cup P_1, {<_0} \cup {<_1}, \Phi_0 \cup \Phi_1)$\@.
The empty pomset $\emptypset = (\emptyset, \emptyset, \emptyset)$ is the unit for sequential and parallel composition.
Given a pomset $P$ on a set of labels $L$ and a subset $L' \subseteq L$, the \defin{restriction} $P\restriction_{L'}$ of $P$ to $L'$ is the pomset on $\Phi^{-1}(L')$ whose ordering is induced by $P$\@.
The \defin{deletion} of $L'$ from $P$ is $P\restriction_{L\setminus L'}$\@.
We lift these operations to sets of pomsets in the obvious manner, \eg, $S_1; S_2 = \{P_1; P_2 \mid P_i \in S_i\}$\@.
\qed%
\end{defi}

\noindent Because the skip action $\delta$ has no effects, we identify program orders $P$ and $P'$ whenever there exists a non-empty pomset $P_\delta$ that can be obtained in two ways: by deleting a finite number of $\delta$ actions from $P$ and also by deleting a finite number of $\delta$ actions from $P'$\@.
This means, \eg, that we identify $\{\delta\};P$, $\{\delta\}\parop P$, and $P$ whenever $P \neq \emptypset$, but $\{\delta\};\emptypset = \{\delta\} \neq \emptypset$\@.

We begin with the program order denotation of expressions.
To each expression $e$, we assign a set $\Ppo(e)$ of tuples of program orders and corresponding values:
\begin{align*}
  \Ppo &: \VExp \to \wp(\poms(\Apo) \times V)\\
  \Ppo(v) &= \{(\{\delta\}, v)\}\\
  \Ppo(x) &= \{(\{x=v\},v) \mid v \in V \}\\
  \Ppo(e_1 \odot e_2) &= \{(P_1 \parop P_2, v_1 \odot v_2) \mid (P_i, v_i) \in \Ppo(e_i)\}
\end{align*}
where $\odot$ ranges over binary operations.
Read expressions $x$ are associated with arbitrary values in $V$ for reasons of compositionality: we do not know with which writes the read may eventually be composed, and so we need to permit reading arbitrary values.
We chose to evaluate binary operations $e_1 \odot e_2$ in parallel; one could just as legitimately have chosen to sequentialise the evaluation and written $P_1;P_2$\@.
We assume $v_1 \odot v_2 \in V$ to be the result of applying the binary operation $\odot$ to $v_1$ and $v_2$\@.
We handle program orders for unary expressions analogously, and assume $\neg b'$ is the result of negating the boolean value $b'$\@.
To simplify the clauses involving conditionals, we give helper functions $\calP_\truev(b)$ and $\calP_\falsev(b)$ to extract the pomsets corresponding to the given boolean values from $\Ppo(b)$\@.
\begin{align*}
  \Ppo &: \BExp \to \wp(\poms(\Apo) \times \Bool)\\
  \Ppo(b) &= \{ (\{\delta\}, b) \}\qquad (b \in \{\truev,\falsev\})\\
  \Ppo(\neg b) &= \{ (P, \neg b') \mid (P, b') \in \Ppo(e) \}\\
  \Ppo(e_1 \odot e_2) &= \{(P_1 \parop P_2, v_1 \odot v_2) \mid (P_i, v_i) \in \Ppo(e_i)\}\\
  \calP_\truev(b) &= \{ P \mid (P,\truev) \in \Ppo(b) \}\\
  \calP_\falsev(b) &= \{ P \mid (P,\falsev) \in \Ppo(b) \}
\end{align*}
Note that in the case of boolean binary operations, the $e_i$ might be integer or boolean expressions, and the corresponding semantic clause for $\Ppo(e_i)$ should be used.

We give the program order denotation of commands in a similar manner, this time associating sets of program orders to each command phrase:
\begin{align*}
  \Ppo &: \Cmd \to \wp(\poms(\Apo))\\
  \Ppo(\cskip) &= \{\{\delta\}\}\\
  \Ppo(x:=e) &= \{ P;\{x:=v\} \mid (P,v) \in \Ppo(e)\}\\
  \Ppo(c_1;c_2) &= \Ppo(c_1);\Ppo(c_2)\\
  \Ppo(c_1\parop c_2) &= \Ppo(c_1) \parop \Ppo(c_2)\\
  \Ppo(\ifb{b}{c_1}{c_2}) &= \calP_\truev(b);\Ppo(c_1) \cup \calP_\falsev(b);\Ppo(c_2)\\
  \Ppo(\whileb{b}{c}) &= \bigcup_{n=0}^\infty I^{n}(b,c)
                        \cup I^\omega(b,c),
\end{align*}
where $I^0(b,c) = \calP_\falsev(b)$ and $I^{n+1}(b,c) = \calP_\truev(b);\Ppo(c);I^n(b,c)$\@.

The only interesting clause is for $\whileb{b}{c}$\@.
Here, we take union of all of the finite unrollings $I^n(b,c)$ of the loop.
We must also consider the case of an infinite loop.
This is captured by $I^\omega(b,c)$, which describes the infinite pomset obtained by unrolling the loop countably infinitely many times.
The $\whileb{b}{c}$ clause also illustrates why we associate the pomset $\{\delta\}$ instead of $\emptypset$ to values: otherwise, we would have $\Ppo(\whileb{\falsev}{c}) = \{\emptypset\}$, and this would break our intuition that this program should be denotationally equivalent~to~$\cskip$\@.
It would also have no executions under the formal account of Section~\ref{sec:executions}.

To illustrate the above semantic clauses, we return to the Dekker program from the introduction.
This program has pomsets of each of the following forms, for each choice of $v \neq 0$ and $v' \neq 0$:
\[ \xymatrix@C-2em@R-1.5em{x:= 1\ar[d] & y:=1\ar[d]\\
    y = 0\ar[d] & x = 0\ar[d]\\
    z := 1 & w := 1},
  \qquad
  \xymatrix@C-2em@R-1.5em{x:= 1\ar[d] & y:=1\ar[d]\\
    y = v & x = 0\ar[d]\\
    & w := 1},
  \qquad
  \xymatrix@C-2em@R-1.5em{x:= 1\ar[d] & y:=1\ar[d]\\
    y = 0\ar[d] & x = v\\
    z := 1 &},
  \qquad
  \xymatrix@C-2em@R-1.5em{x:= 1\ar[d] & y:=1\ar[d]\\
    y = v & x = v'}.
\]
The first program order describes an execution where we read both $y=0$ and $x=0$ and where Dekker fails.
The next three forms of program order describe executions in which one or both reads obtain a non-zero~value.

\subsection{TSO Pomsets}%
\label{subsec:tsopoms}

In this subsection, we assign a set $\PTSO(p)$ of \textit{TSO pomsets} to each program $p$, serving as the abstract meaning or denotation of the program $p$ under the TSO memory model.
To do so, we will need to carefully model write buffers.
For compactness, we will write $\PT$ instead of $\PTSO$ in this section's semantic clauses.

We introduce a set $\BLoc = \{ \bar x \mid x \in \Loc \}$ of buffer locations, assumed to be in bijection with $\Loc$, and let the set of buffer write actions be $\Ab = \{ {\bar x}~:=~v \mid {\bar x}~\in~\BLoc,\ v~\in~ V \}$\@.
An action ${\bar x}~:=~v$ by a thread denotes adding a write $x~:=~v$ to the thread's write buffer.
The set of TSO actions $\Atso$ then consists of $\Apo$ extended with $\Ab$\@.
A \textit{TSO pomset} will then be a pomset in $\poms(\Atso)$ satisfying the finite height property.

To capture the effects of buffers, we parametrize our semantic clauses with lists of global write actions, which represent the writes currently in our buffer.
We let $\Ls = \lists{\Aw}$ be the set of all lists.
The intuition is that write buffers behave as queues under TSO, and we can use a list $L \in \Ls$ to model a queue by dequeuing from the head of the list and enqueuing at the end of the list.
For expository convenience, we identify lists and \textit{linear} pomsets, where we say a pomset is linear if its underlying poset is linear.
Explicitly, we identify $\emptylist$ with the empty pomset $\emptypset$, and $\lst{\lambda_1, \dotsc, \lambda_n}$ with the pomset $\{\lambda_1\};\cdots;\{\lambda_n\}$\@. 
To minimize notation, we leverage this identification and write $L;L'$ to denote the concatenation of $L$ and $L'$\@.

The semantic clauses are given in two strata.
The semantic clauses $\PB$ for ``basic TSO pomsets'' capture the meaning of the syntactic phrases in a manner very similar to the program order definitions in Section~\ref{subsec:popoms}.
$\PB$ assign to each command phrase a function from buffer lists to a set of pairs of TSO pomsets and buffer lists.
We present these clauses using the abbreviation $\PB_L(c) = \PB(c)(L)$\@.
The pomset component of $\PB_L(c)$ captures the meaning of the phrase in the presence of the buffer $L$, while the buffer component describes the state of the buffer after performing the actions associated with the phrase.
In the second stratum, we use $\PT$ clauses to capture the meaning of the phrase in the presence of buffer flushing.
Flushing a write from a buffer $L$ consists of dequeuing a global write $x:=v$ from $L$ and inserting it in the pomset.
$\PT_L(c) = \PT(c)(L)$ is again a subset of $\poms(\Atso) \times \Ls$\@.

To generate TSO pomsets, we modify the semantic clauses generating program orders in four key places to get our basic pomsets.
The first is for write commands $x:=e$\@.
Starting from a buffer $L \in \Ls$, we get the pomset $P$ and associated value $v$ for $e$ from the denotation $\PT_L(e)$ instead of $\Ppo(e)$\@.
The buffer $L$ may have changed to a buffer $L'$ while we were evaluating $e$, and $\PT_L(e)$ also gives us this $L'$\@.
Instead of immediately making a global write to $x$ as we would have in the program order clause, we enqueue the global write on the buffer $L'$:
\[ \PB_L(x:=e) = \{(P;\{\bar x:=v\}, L'; \{x:=v\}) \mid (P,v,L') \in \PT_L(e)\}. \]

We must also change the semantic clauses for read expressions.
By axiom \rn{Vb}, whenever we read from a location $x$, we must use the most recent value available for it in the write buffer, if available.
We use the following helper function to convert a buffer $L \in \Ls$ to a partial function $\lsState{L} : \Loc \pfin V$ giving us the value of the most recent write in $L$ to a given location:
\begin{align}
  \lsState{\emptylist}(x) &= \mathit{undefined,} &
  \lsState{L;\{x:=v\}}(y) &= \begin{cases}
    v & \text{if } x = y\\
    \lsState{L}(y) & \text{otherwise}.
  \end{cases}\label{eq:6}
\end{align}
Then, the semantic clause giving us the basic pomsets for reads is
\[ \PB_L(x) = \{(\{x=v\},v,L) \mid \lsState{L}(x) = v \} \cup \{(\{x=v\},v,L) \mid x \notin \dom(\lsState L), v \in V \}.
\] The first part tells us to use the value associated with $x$ in the buffer $L$, if available.
The second part uses arbitrary values if the value is unavailable, as with program orders.

The third major change involves parallel composition.
We explain parallel composition of expressions; parallel composition of commands is analogous.
By axioms \rn{F} and \rn{J}, we must flush our buffers before every fork and join.
We therefore begin by flushing our entire buffer, \ie, by taking $L$ and placing it at the beginning of our pomset.
Having flushed the buffer, we then evaluate the $e_i$ with empty buffers and get back pomsets $P_i$ and $v_i$\@.
Because we can only join threads if their buffers are empty, we require that these $P_i$ and $v_i$ be associated with empty buffers in $\PT_\emptylist(e_i)$\@.
We then proceed as for the program order, and add the parallel composition of the $P_i$ to our pomset, and compute the value $v_1 \odot v_2$\@.
Because we just joined two empty buffers, our resulting buffer is empty:
\[ \PB_L(e_1 \odot e_2) = \{(L;(P_1 \parop P_2), v_1 \odot v_2, \emptylist) \mid (P_i, v_i, \emptylist) \in \PT_\emptylist(e_i)\}.\]

Finally, when we sequentially compose two commands $c_1$ and $c_2$ (assuming no forking or joining), $c_2$ continues executing from the buffer $c_1$ finished with.
$\PT(c_1)$ and $\PT(c_2)$ are both functions of type $\Ls \to \wp(\poms(\Atso)\times\Ls)$ and are not composable \textit{qua} functions.
Consequently, we need to define a composition operation capturing the above the operational intuition.
This composition is the polymorphic function $\seqit$, where $S \in \wp(\poms\times A)$ and $f \in A \to \wp(\poms\times B)$:
\begin{align*}
  &{\seqit} : \forall A . \forall B. \wp(\poms\times A) \to (A \to \wp(\poms\times B)) \to \wp(\poms\times B)\\
  S &\seqit f = \{ (P;P', b) \mid {(P, a) \in S} \aand {(P', b) \in f(a)} \}.
\end{align*}
Taking $A = B = \Ls$, sequential composition can be expressed using as \[
  \PB_L(c_1;c_2) = \PT_{L}(c_1) \seqit \PT(c_2).
\]
Explicitly, this means $\PB_L(c_1;c_2) = \{(P_1;P_2, L_2) \mid (P_1,L_1) \in \PT_L(c_1), (P_2, L_2) \in \PT_{L_1}(c_2) \}$\@.
This idiom of chaining pairs of pomsets and buffers together using $\seqit$ will be useful throughout.
We make $\seqit$ polymorphic so that we can handle, \eg, the case of $A = \Ls$ and $B = V \times \Ls$ below.

The remainder of the basic clauses are analogous to those for program order pomsets, subject to the modifications described above:
\begin{align*}
  \PB &: \VExp \to \Ls \to \wp(\poms(\Atso) \times V \times \Ls)\\
  \PB_L(v) &= \{(\{\delta\}, v, L)\}\\
  \PB &: \BExp \to \Ls \to \wp(\poms(\Atso) \times \Bool \times \Ls)\\
  \PB_L(\neg e) &= \{ (P, \neg b, L') \mid (P, b, L') \in \PB_{L}(e)
                  \}\\
  \PT_{L,\truev}(b) &= \{ (P, L') \mid (P,\truev,L') \in \PT_L(b) \}\\
  \PT_{L,\falsev}(b) &= \{ (P, L') \mid (P,\falsev,L') \in \PT_L(b) \}\\
  \PB &: \Cmd \to \Ls \to \wp(\poms(\Atso) \times \Ls)\\
  \PB_L(\cskip) &= \{(\{\delta\}, L)\}\\
  \PB_L(c_1\parop c_2) &= \{(L;(P_1 \parop P_2), \emptylist) \mid (P_i, \emptylist) \in \PT_\emptylist(c_i) \}\\
  \PB_L(\ifb{b}{c_1}{c_2}) &= \left(\PT_{L,\truev}(b) \seqit \PT(c_1)\right) \cup \left(\PT_{L,\falsev}(b) \seqit \PT(c_2)\right)\\
  \PB_L(\whileb{b}{c}) &= \bigcup_{n=0}^\infty I_L^{n}(b,c) \cup I_L^\omega(b,c)
\end{align*}
where $I_L^0(b,c) = \PT_{L,\falsev}(b)$ and $I_L^{n+1}(b,c) = \PT_{L,\truev}(b) \seqit \PT(c) \seqit I^n(b,c)$\@.
The set $I^\omega_L(b,c) \subseteq \poms(\Atso)\times\{\emptylist\}$ contains all infinite pomsets obtained through countably infinitely many unfoldings; because we can never observe the buffer at the end, we treat it as empty to simplify presentation.

We now turn our attention to flushing.
The intent is that a thread can flush arbitrarily many of its writes at any point in its execution.
Thus, the pomsets associated with flushes for a buffer $L$ are the prefixes $L'$ of $L$, and the resulting buffers are the remainders of $L$\@.
We use $\splt(L)$ to denote these prefix-suffix pairs:
\begin{align*}
  &\splt : \Ls \to \wp(\poms(\Atso)\times\Ls)\\
  &\splt(L) = \{(L',L'') \mid L = L'; L''\}
\end{align*}

We introduce a variant of $\seqit$ to cope with triples of pomsets, values, and buffers, and will rely on types to disambiguate the version needed in any given situation:
\begin{align*}
  &{\seqit} : \forall A . \forall B. \wp(\poms\times A \times B) \to (B \to \wp(\poms\times B)) \to \wp(\poms\times A \times B)\\
  S &\seqit f = \{ (P;P', A, B') \mid {(P, A, B) \in S} \aand {(P', B') \in f(B)} \}
\end{align*}

We define the TSO pomsets $\PT$ in terms of $\PB$ and $\splt$\@.
$\PT$ composes $\splt$ and $\PB$ in a manner that we can flush some writes from the buffer, then evaluate $e$ or perform $c$, and then flush some writes at the end:
\begin{align*}
  \PT_L(e) &= \splt(L) \seqit \PB(e) \seqit \splt\\
  \PT_L(c) &= \splt(L) \seqit \PB(c) \seqit \splt
\end{align*}
We can validate various expected equivalences by unfolding these definitions.
For example, sequential composition of commands is associative, because
\[
  \PT_L(c_1;(c_2;c_3)) = \splt(L) \seqit \penalty 0 \PB(c_1) \seqit \splt \seqit \PB(c_2) \seqit \splt \seqit \PB(c_3) \seqit \splt = \PT_L((c_1;c_2);c_3).
\]
Using the identity $\emptypset;P = P$ and the fact that parallel composition of pomsets is associative, one can show that parallel composition of commands is associative, \ie, that $\PT_L(c_1 \parop\, (c_2\parop c_3)) = \PT_L((c_1 \parop c_2) \parop c_3)$\@.
The parallel composition of pomsets commutes, so the parallel composition of commands commutes, \ie, $\PT_L(c_1 \parop c_2) = \PT_L(c_2 \parop c_1)$\@.

To illustrate the effects of flushing and the effect of buffers on reads, we consider the expression $x$ in the presence of the buffer $L = \lst{x:=3,y:=2}$\@.
The triples $(P, v, L') \in \PT_L(x)$ are of the form
\begin{align*}
  P &= (\xymatrix@C-1.5em{x=3}) & v &= 3 & L' &= \lst{x:=3,y:=2}\\
  P &= (\xymatrix@C-1.5em{x=3 \ar[r] & x := 3}) & v &= 3 & L' &= \lst{y:=2}\\
  P &= (\xymatrix@C-1.5em{x=3 \ar[r] & x := 3 \ar[r] & y:=2}) & v &= 3 & L' &= \emptylist\\
  P &= (\xymatrix@C-1.5em{x:=3 \ar[r] & x=v}) & v &\in V & L' &= \lst{y:=2}\\
  P &= (\xymatrix@C-1.5em{x:=3 \ar[r] & x=v \ar[r] & y:=2}) & v &\in V & L' &= \emptylist\\
  P &= (\xymatrix@C-1.5em{x:=3 \ar[r] & y:=2 \ar[r] & x=v}) & v &\in V & L' &= \emptylist
\end{align*}
In the first case, the resulting $P$ denotes performing the read without also doing any flushing.
In this case, because we have a write to $x$ in the buffer, the read from $x$ must use its value.
The second and third case give rise to pomsets denoting performing the read before doing one or two flushes.
In the fourth and fifth cases, we first flush the write to $x$ from the buffer; the value $v$ read from $x$ is then free to range over all possible values because there are no other writes to $x$ in the buffer.
Finally, in the last case, we flush all of the writes from the buffer before reading $x$, and the value read from $x$ is again unconstrained.

\begin{defi}
  The \defin{TSO pomsets} of a program $p$ are $\PTSO(p) = \{ P \mid (P, \emptylist) \in \PT_\emptylist(p) \}$.
  \qed%
\end{defi}
\noindent This definition selects pomsets with empty buffers from $\PT_\emptylist(p)$ because we expect programs to be run from an empty buffer and to only stop after emptying all buffers.

We illustrate the constructions by giving four example families of TSO
pomsets for the Dekker program from the introduction, again assuming
$v \neq 0$ and $v' \neq 0$:
\[ \xymatrix@C-2em@R-1.5em{\bar x:= 1\ar[d] & \bar y:=1\ar[d]\\
    x:=1\ar[d] & y:=1\ar[d]\\
    y = 0\ar[d] & x = 0\ar[d]\\
    \bar z := 1\ar[d] & \bar w := 1\ar[d]\\
    z := 1 & w := 1},
  \qquad
  \xymatrix@C-2em@R-1.5em{\bar x:= 1\ar[d] & \bar y:=1\ar[d]\\
    y = 0\ar[d] & x = 0\ar[d]\\
    x:=1\ar[d] & y:=1\ar[d]\\
    \bar z := 1\ar[d] & \bar w := 1\ar[d]\\
    z := 1 & w := 1},
  \qquad
  \xymatrix@C-2em@R-1.5em{\bar x:= 1\ar[d] & \bar y:=1\ar[d]\\
    x:=1\ar[d] & y:=1\ar[d]\\
    y = v & x = 0\ar[d]\\
    & \bar w := 1\ar[d]\\
    & w := 1},
  \qquad
  \xymatrix@C-2em@R-1.5em{\bar x:= 1\ar[d] & \bar y:=1\ar[d]\\
    y =v\ar[d] & y:=1\ar[d]\\
    x := 1 & x = v'}.
  \label{dekker-tso-poms}
\]
In the first family of pomsets, we flush the writes immediately after inserting them in the buffers, while in the second, we flush the writes to $x$ and $y$ after reading $y$ and $x$\@.
In the third family, we flush $x$ right after placing its write in the buffer, but fall into the $\falsev$ case of the conditional after reading some value $v \neq 0$, thus taking the $\cskip$ branch.
In the fourth pomset, we read $y$ after placing the write $x:=1$ in the buffer, but before it gets flushed, and both threads fall into the $\cskip$ branch.

\subsubsection{Laws of parallel programming}

Because parallel composition of pomsets is associative and commutative, and because sequential composition of pomsets is associative, we satisfy many laws of parallel programming from~\cite{Brookes1996} and~\cite{Jagadeesan2012} for free.
Let $C \equiv C'$ whenever $\PT_L(C) = \PT_L(C')$ for all $L$\@.
Our semantics satisfies:
\begin{align*}
  \cskip; c &\equiv c \equiv c; \cskip\\
  (c_1;c_2);c_3 &\equiv c_1;(c_2;c_3)\\
  c_1 \parop c_2 &\equiv c_2 \parop c_1\\
  (c_1 \parop c_2) \parop c_3 & \equiv c_1 \parop (c_2 \parop c_3)\\
  (\ifb{b}{c_1}{c_2});c_3 &\equiv \ifb{b}{c_1;c_3}{c_2;c_3}\\
  \whileb{b}{c} &\equiv \ifb{b}{(c;\whileb{b}{c})}{\cskip}
\end{align*}
Because we must flush buffers before every fork and join, $\cskip$ is not a unit for parallel composition under TSO\@: $\cskip \parop c \not\equiv c$\@.

\subsection{Executions}%
\label{sec:executions}

Our TSO pomset semantics gives an abstract account capturing families of possible executions.
However, compositionality comes with its price: we associate to programs some pomsets that cannot in any real sense be ``executed''.
Consider for example, the pomset $\bar x:=2 \to x := 2 \to x = 1 \to \cdots$ for the program $c~=~(x~:=~2;~\ifb{~x~=~1~}{c_1}{c_2})$\@.
In no circumstance do we expect to execute $c_1$ when this program is run alone, and so the above TSO pomset has, in a sense made precise later, no executional meaning.
However, compositionality requires this pomset be associated with the command $c$, because one \textit{could} execute $c_1$ if our program were instead $c \parop x:= 1$\@.
Our notion of \textit{execution} filters out these pomsets with no executional meaning and yields an input/output behaviour for programs built from their pomset semantics.

\subsubsection{Buffered States}

Our notion of execution uses buffered global state, \ie, a global state with a write buffer per thread.
We execute threads individually.
Each thread's execution starts from a state with a buffer, which in combination reflect that thread's view of shared memory.
Let $\Locs = \BLoc \cup \Loc$ be the set of all locations.
We use elements of $\Sigma = (\BLoc \penalty 1000 \pfin \penalty 1000 (V\times\N)/{\approx}) \times (\Loc \penalty 1000 \pfin \penalty 1000 V)$ to model the combination of a global state and a buffer, where $\approx$ is given by $(v, n) \approx (v', m)$ if and only if both $n = m$, and $v = v'$ or $n = 0$\@.
Because $\Loc$ and $\BLoc$ are disjoint, we identify $\Sigma$ with its obvious inclusion in $\Locs \pfin ({(V\times\N)/{\approx}} \cup V)$\@.
The intuition is that if $\sigma(\bar x) = (v, n)$, then there are $n$ writes to $x$ in $\sigma$'s write buffer, and the most recent buffer write to $\sigma$ had the value $v$\@.
We need to keep track of the number $n$ of writes to $x$ still in the buffer to know whether we should continue reading $x$ from the buffer after a flush.
We identify $(v,0)$ and $(v',0)$ for all $v$ and $v'$ because one should not be able to observe a value for a write that is no longer in the buffer, and this identification allows us to ``forget'' the value by setting $n$ to $0$\@.
For $x, y, z, \dotsc \in \Locs$ and $u, v, w, \dotsc$ in the corresponding subset of $((V\times\N)/\approx) \cup V$, we denote by $\state{x:u,y:v,z:w,\dotsc}$ the buffered state with graph $\{(x,u),(y,v),(z,w),\dotsc\}$\@.
For compactness of notation, we write $v_n$ for the equivalence class of $(v, n)$ in $(V\times\N)/\approx$\@.

\subsubsection{Footprints}%
\label{sec:footprints}

Footprints are the first step towards formalizing execution and filtering out unexecutable pomsets.
Informally, a footstep $(\sigma, \tau) \in \Sigma\times\Sigma$ of an action $\lambda$ is a minimal piece of state $\sigma$ required to be able to perform $\lambda$, and a description $\tau$ of the effects of performing $\lambda$\@.
For example, to perform a global write $x:=v$, $x$ must be in the domain of the initial state and present in the buffer, so $\sigma = [ x:v', \bar x:v''_{n+1}]$ for some $v'$ and $v''$, and the result is setting the global value of $x$ to $v$ while removing one occurrence of $x$ from the buffer, so $\tau = [x:v, \bar x:v''_{n}]$\@.
Though $v$ and $v''$ are unrelated, this gives the correct behaviour in the context of command pomsets because global writes to $x$ occur in the same order as buffer writes to $x$\@.
To perform a read action $x=v$, we must either have no entries for $x$ in the buffer and have $x:v$ in the global state, or we must have $x$ in the buffer with value $v$, \ie, $\bar x:v_n$ for some $n > 0$\@.
We call the set of footsteps associated with an action its \textit{footprint}.
Pomsets also have footsteps and footprints.

\begin{defi}
\defin{TSO footprints} for actions are given as follows:
\begin{align*}
  \footp{x=v} &= \{ ( [ x : v, \bar x : v_0], \emptystate), ([\bar x : v_{n+1}], \emptystate) \mid n \in \N \}\\
  \footp{\bar x := v} &= \{ ([\bar x : v'_n], [\bar x:v_{n+1}]) \mid v' \in V \aand n \in \N \}\\
  \footp{x:=v} &= \{([x:v',\bar x:v''_{n+1}],[x:v, \bar x : v''_n])\mid v', v''\in V \aand n \in \N\}\\
  \footp{\delta} &= \{ (\emptystate, \emptystate) \}
  \tag*{\qed}
\end{align*}
\end{defi}

\noindent To give footprints to pomsets, we need to know when it makes sense to combine two footsteps sequentially or in parallel.
We say two states $\sigma_1$ and $\sigma_2$ are \defin{consistent}, $\scons{\sigma_1}{\sigma_2}$, if for all $x \in \dom(\sigma_1) \cap \dom(\sigma_2)$, $\sigma_1(x) = \sigma_2(x)$\@.
In this case, $\sigma_1 \cup \sigma_2$ is also a state.
Let $\sigma \setminus \dom(\tau) = \sigma \restriction_{\dom(\sigma) \setminus \dom(\tau)}$\@.
Then the result of updating $\sigma$ by $\tau$ is $\supd{\sigma}{\tau} = (\sigma\setminus\dom(\tau)) \cup \tau$\@.
To sequence the footprint $(\sigma_1, \tau_1)$ before the footprint $(\sigma_2, \tau_2)$, we must ensure that the result of the first computation from its initial state, $\supd{\sigma_1}{\tau_1}$, is consistent with the requirements $\sigma_2$ of the second computation, \ie, we must ensure $\scons{\supd{\sigma_1}{\tau_1}}{\sigma_2}$\@.
In this case, sequentially performing both computations requires that the initial state provide everything required by the first computation, plus everything required by the second computation not already provided by the first one.
This is $\sigma_1 \cup (\sigma_2 \setminus  \dom(\tau_1))$\@.
The effects of the two computations are the first's effects updated by the effects of the second one, $\supd{\tau_1}{\tau_2}$\@.
Consequently, for states $(\sigma_i, \tau_i)$ such that $\scons{\supd{\sigma_1}{\tau_1}}{\sigma_2}$, we define their sequential composition to be:
\[
  (\sigma_1, \tau_1) \seqop (\sigma_2, \tau_2) = (\sigma_1 \cup (\sigma_2 \setminus  \dom(\tau_1)), \supd{\tau_1}{\tau_2}).
\]
This lifts to an associative operation on $\Sigma \times \Sigma$:
\[
  S_1 \seqop S_2 = \{ (\sigma_1, \tau_1) \seqop (\sigma_2, \tau_2) \mid (\sigma_i, \tau_i) \in S_i \aand \scons{\supd{\sigma_1}{\tau_1}}{\sigma_2} \}.
\]

To account for global writes occurring elsewhere during the program, we parametrize the rules assigning footprints $\footp{P}_\Lambda$ to pomsets $P$ by a list $\Lambda$ containing a linearisation of the pomset as a subsequence, combined with any number of other global writes that represent flushes from buffers belonging to other threads.
Because these global writes are performed, in principle, by other threads, they should not affect the buffers in the footprints associated with $P$\@.
The exact mechanics of how $\Lambda$ and $P$ interact in producing footprints will be made clear below.
Formally, given a pomset $P$, we let $\Lin(P)$ be the set of its linearisations.
Then the clauses are parametrized by $\Lambda \in \Gamma(P) = \bigcup_{L \in \Ls} \Lin(P \parop L)$, where given some global-write environment $\Lambda \in \Gamma(P)$, we identify $P$ with its image in $\Lambda$\@.
So $\Lambda$ is a linearisation of $P$ interspersed with global writes that will not involve buffers.
In particular, when $P$ is $\{\lambda\}$, then every $\Lambda \in \Gamma(P)$ will be of the form $\Lambda_1;\{\lambda\};\Lambda_2$ for some unique $\Lambda_1, \Lambda_2 \in \Ls$\@.
This fact will be important for understanding the \rn{Act} rule below.
Because $\Lambda$ is a linearisation, it totally orders the writes it contains, including those of $P$\@.

Given some $L\in \Ls$, let $\footp{L}^*$ be inductively defined on the structure of $L$ as follows:
\begin{align*}
  \footp{\,\emptylist\,}^* &= \{(\emptystate, \emptystate)\}\\
  \footp{\{x:=v\} ; L}^* &= \{ ([x:v'],[x:v]) \mid v'\in V\} \seqop \footp{L}^*
\end{align*}
The intuition here is that the foreign buffer flushes in $L$ should only affect the global part of the state and have no effect on our buffer.
Given a list $\lst{x_1:=v_1, \dotsc, x_n:=v_n}$, we write $\footp{x_1:=v_1,\dotsc,x_n:=v_n}^*$ for $\footp{\lst{x_1:=v_1,\dotsc,x_n:=v_n}}^*$\@.

We will need to know if buffered states have empty buffers.
We let $\zeta(\sigma)$ hold if and only if for all $x \in \dom(\sigma \restriction_\BLoc)$, $\sigma(x) = v_0$\@.
In other words, $\zeta(\sigma)$ holds if and only if $\sigma$ ``has an empty buffer''.
Thus, for example, both $\zeta(\emptystate)$ and $\zeta([x:1,\bar y:2_0])$, but neither $\zeta([x:1, \bar x :2_5])$ nor $\zeta([\bar y:1_1])$\@.

\begin{defi}
  The \defin{footprint} $\footp{P}_\Lambda$ of a pomset $P$ under a global-write environment $\Lambda \in \Gamma(P)$ is the smallest set closed under the following three rules:
  \begin{description}[before={\renewcommand\makelabel[1]{##1}}]
  \item[\rn{Act}] If $P = \{\lambda\}$ for some action $\lambda$ and $\Lambda = \Lambda_1;P;\Lambda_2$ for some $\Lambda_1, \Lambda_2 \in \Ls$, then $\footp{P}_\Lambda = \footp{\Lambda_1}^* \seqop \footp{\lambda} \seqop \footp{\Lambda_2}^*$\@.
  \item[\rn{Seq}] If $P = P_1; P_2$ and $\Lambda = \Lambda_1;\Lambda_2$, then $\footp{P_1}_{\Lambda_1} \seqop \footp{P_2}_{\Lambda_2} \subseteq \pex{P}_\Lambda$\@.
  \item[\rn{Par}] If $P = P_1 \parop P_2$, $\Lambda_1$ is the result of deleting the read and buffer write actions of $P_2$ from $\Lambda$, $\Lambda_2$ is the symmetric restriction, $(\sigma_i, \tau_i) \in \pex{P_i}_{\Lambda_i}$, $\zeta(\sigma_i)$ and $\zeta(\tau_i)$ ($i = 1,2$), and $\sigma_1 \Uparrow \sigma_2$, then $(\sigma_1 \cup \sigma_2, \tau_1 \cup \tau_2) \in \pex{P}_\Lambda$\@.
  \qed%
  \end{description}
\end{defi}

\noindent This definition is inspired by Lamport's ``happened before'' relation~\cite{Lamport1978}.
In the case of \rn{Act}, for a given $P = \{\lambda\}$ and $\Lambda \in \Gamma(P)$, the intuition is that $\Lambda$ specifies that the global writes in $\Lambda_1$ happened before $\lambda$, and that $\lambda$ happened before the global writes in $\Lambda_2$\@.
In the case of \rn{Seq}, for $\footp{P_i}_{\Lambda_i}$ to be well-defined, we are implicitly assuming that $\Lambda_i \in \Gamma(P_i)$\@.
\rn{Seq} tells us that the result of sequentially executing a program in the presence of global writes should be the same as executing the pieces sequentially in the presence of the appropriate subset of global writes.
Finally, in \rn{Par}, the restrictions of $\Lambda$ are such that both parallel components observe \textit{all} writes in the same order, and this is how we simulate the effects of writes to a global state.
One can show that the set $\footp{P}_\Lambda$ is well-defined with regards to our identification of pomsets up-to-isomorphism.
In particular, $\footp{(P_1;P_2);P_3}_\Lambda = \footp{P_1;(P_2;P_3)}_\Lambda$ and $\footp{(P_1\parop P_2)\parop P_3}_\Lambda = \footp{P_1\parop(P_2 \parop P_3)}_\Lambda$\@.

We give various results below that simplify computing the footprint of pomsets.%
\label{pars:sp-explan}
By Proposition~\ref{prop:6}, a pomset has a footprint only if it is series-parallel.
A pomset is \defin{series-parallel} or \defin{SP} if it is linear, or if it is the sequential or parallel composition of SP pomsets.

The maximal linear segments of a SP pomset are its \textit{SP components}; any SP pomset can be uniquely decomposed into these.
To compute the footprint of a pomset $P$, we can use Proposition~\ref{prop:2} and~\ref{prop:8} to first decompose $P$ into its series-parallel components, each of which will be linear.
We can then compute the footprints of these linear components using Corollary~\ref{cor:1} and Proposition~\ref{prop:12}, and combine them using the appropriate applications of \rn{Seq} and~\rn{Par}.

We use these results to illustrate how global-write environments simulate the effects of global writes.
Let $P_1$ and $P_2$ be the pomsets
\[ \xymatrix{\bar x :=  2 \ar[r] & x:=2 \ar[r] & x = 3} \qquad\text{and}\qquad \xymatrix{\bar x :=  3 \ar[r] & x=3 \ar[r] & x:=3},
\] respectively, and let $P = P_1 \parop P_2$\@.
Consider the global-write environments
\begin{align*}
  \Lambda &= \lst{\bar x := 2, \bar x := 3, x := 2, x = 3, x:= 3, x = 3}\\
  \Lambda_1 &= \lst{\bar x := 2, x := 2, x := 3, x = 3}\\
  \Lambda_2 &= \lst{\bar x := 3, x := 2, x = 3, x:= 3}.
\end{align*}
We see that $\Lambda$ is a global-write environment for $P$, and that $\Lambda_i$ is the restriction of $\Lambda$ given by the \rn{Par} rule for $P_i$ and is again a global-write environment for $P_i$\@.
To compute the footprint of $P$, we must begin by applying the \rn{Par} rule and recursively compute the footprint of $P_i$ under $\Lambda_i$\@.
We first consider $\footp{P_1}_{\Lambda_1}$\@.
By Corollary~\ref{cor:1}, it is given by \[
  \footp{\bar x := 2} \seqop \footp{x:=2} \seqop \footp{x:=3}^* \seqop \footp{x=3},
\] where we omitted the instances of the unit $\footp{\emptylist}^*$ for the $\seqop$ operation.
Simplifying this expression, we get that $\footp{P_1}_{\Lambda_1} = \{(\state{x:v, \bar x:v'_n}, \state{x:3, \bar x:2_n}) \mid v, v' \in V \aand n \in \N \}$\@.
Despite there being no write of $3$ to $x$ in $P_1$, the presence of such a global write in $\Lambda_1$ and its position in $\Lambda_1$ mean that the read $x=3$ can be executed by $P_1$\@.
More interesting, perhaps, is the footprint $\footp{P_2}_{\Lambda_2}$, and in particular, how the global write $x:=2$ does not interfere with the read $x = 3$ thanks to buffering.
Again by Corollary~\ref{cor:1}, it is given by \[
  \footp{\bar x := 3} \seqop \footp{x:=2}^* \seqop \footp{x=3} \seqop \footp{x:=3}.
\]
The subexpression $\footp{\bar x := 3} \seqop \footp{x:=2}^*$ simplifies to
\begin{equation}
  \{ (\state{x:v,\bar x:v'_n}, \state{x:2,\bar x:3_{n+1}}) \mid v, v' \in V \aand n \in \N \}\label{eq:1}
\end{equation}
Because for each $(\sigma, \tau) \in \footp{\bar x := 3} \seqop \footp{x:=2}^*$ we have $\tau(\bar x) = 3_{n}$ for some $n > 0$, \ie, because $\tau$'s buffer has a write to $x$, the footsteps $(\state{x:3,\bar x:v_0}, \emptystate) \in \footp{x = 3}$ describing a read of $3$ from shared memory are ignored when combining $\footp{\bar x := 3} \seqop \footp{x:=2}^*$ with $\footp{x=3}$ using the $\seqop$ operation.
Instead, $\footp{\bar x := 3} \seqop \footp{x:=2}^*$ gets combined with the subset $\{ (\state{\bar x : 3_{n+1}}, \emptystate) \mid n \in \N \} \subseteq \footp{x=3}$ to give~\eqref{eq:1} again.
Then combining~\eqref{eq:1} with $\footp{x:=3}$ gives us $\footp{P_2}_{\Lambda_2} = \{ (\state{x:v,\bar x:v'_n}, \state{x:3,\bar x:3_{n}}) \mid v, v' \in V \aand n \in \N \}$\@.
Finally, we get $\footp{P}_{\Lambda} = \{ (\state{x:v,\bar x:v_0}, \state{x:3,\bar x:v_0}) \mid v \in V \}$ by combining these footprints using the \rn{Par} rule.

One can show by induction on the rules defining footprints that:

\begin{prop}%
  \label{prop:6}
  For all $P \in \poms(\Atso)$ and $\Lambda \in \Gamma(P)$, if $\footp{P}_\Lambda \neq \emptyset$, then $P$ is SP\@.
  \qed%
\end{prop}

\noindent The following proposition shows that the execution of a pomset $P$ in a global-write environment $\Lambda$ can be thought of as executions of parts of $P$ interspersed with environment steps from $\Lambda$\@.
Its corollary shows how to efficiently compute the footprint of linear pomsets.

\begin{prop}\label{prop:2}
  If $P = P_1;\dotsc;P_n$ and $\Lambda = E_0;\Lambda_1;E_1;\dotsb;\Lambda_n;E_n \in \Gamma(P)$ such that $\Lambda_i \in \Gamma(P_i)$ and $E_i \in \Ls$ for all $i$, then \[
    \footp{P}_\Lambda = \footp{E_0}^* \seqop \footp{P_1}_{\Lambda_1} \seqop \footp{E_1}^* \seqop \dotsb \seqop \footp{P_n}_{\Lambda_n} \seqop \footp{E_n}^*.
  \]
  Without loss of generality, the $\Lambda_i$ and $E_i$ can be chosen such that the first and last elements of $\Lambda_i$ are in $P_i$\@.
  \qed%
\end{prop}

\begin{cor}\label{cor:1}
  If $P = \lst{\lambda_1,\dotsc,\lambda_n}$ and $\Lambda = E_0;\{\lambda_1\};E_1;\dotsb;\{\lambda_n\};E_n \in \Gamma(P)$, then
  \begin{equation*}
    \footp{P}_\Lambda = \footp{E_0}^* \seqop \footp{\lambda_1} \seqop \footp{E_1}^* \seqop \dotsb \seqop \footp{\lambda_n} \seqop \footp{E_n}^*.
    \tag*{\qed}
  \end{equation*}
\end{cor}

\noindent We can similarly decompose parallel compositions of pomsets when computing footprints.

\begin{prop}\label{prop:8}
  If $P = P_1 \parop \dotsb \parop P_n$, $\Lambda \in \Gamma(P)$, and for all $i$, $\Lambda_i \in \Gamma(P_i)$ is obtained from $\Lambda$ by deleting the read and buffer write actions not in $P_i$, then
  \[ \footp{P}_\Lambda = \footp{P_1}_{\Lambda_1} \parop \dotsb \parop\, \footp{P_n}_{\Lambda_n}, \]
  where the associative binary operator $\parop$ on subsets of $\Sigma\times\Sigma$ is given by
  \begin{equation*}
  S_1 \parop S_2 = \{ (\sigma_1 \cup \sigma_2, \tau_1 \cup \tau_2) \mid (\sigma_i, \tau_i) \in S_i, \zeta(\sigma_i), \zeta(\tau_i), \scons{\sigma_1}{\sigma_2}, \scons{\tau_1}{\tau_2} \}.
  \tag*{\qed}
  \end{equation*}
\end{prop}

\noindent Finally, we can characterize the footprints of the above ``environment steps'' using the following lemma, which can be shown by induction on the length of $\Lambda$:

\begin{prop}%
  \label{prop:12}
  For all $\Lambda \in \Ls$ and $(\sigma, \tau) \in \footp{\Lambda}^*$, we have $x := v \in \Lambda$ for some $v$ if and only if $x \in \dom(\sigma)$\@.
  For no $\bar x \in \BLoc$ do we have $\bar x \in \dom(\sigma) \cup \dom(\tau)$\@.
  Moreover, $\dom(\sigma) = \dom(\tau)$\@.
  For all $x \in \dom(\sigma)$ and $v \in V$, we have $(\supd{\sigma}{x:v}, \tau) \in \footp{\Lambda}^*$\@.
  If $x := v$ is the maximal write to $x$ in $\Lambda$, then $\tau(x) = v$\@.
  \qed%
\end{prop}

To better understand and characterize pomset executions, we introduce the \defin{characteristic} $\charac{x}{P} \in \N \times \N$ of a location $x$ in a series-parallel pomset $P$\@.
Intuitively, when $\charac{x}{P} = (g, b)$, $g$ is the number of global writes to $x$ in $P$ that do not have a buffer entry witnessing them, and $b$ describes the number of buffer entries left to be flushed after executing $P$\@.
The characteristic is inductively given on the structure of $P$ as follows:
\begin{align*}
  \charac{x}{\{y:=v\}} &=
                         \begin{cases}
                           (1,0) & x = y\\
                           (0,0) & x \neq y
                         \end{cases}\\
  \charac{x}{\{\bar y := v\}} &=
                                \begin{cases}
                                  (0,1) & x = y\\
                                  (0,0) & x \neq y
                                \end{cases}\\
  \charac{x}{y = v} &= (0,0)\\
  \charac{x}{\delta} &= (0,0)\\
  \charac{x}{P_1 \parop P_2} &= (g_1 + g_2, b_1 + b_2)\\
  \charac{x}{P_1; P_2} &= (g_1 + \maxo{g_2 - b_1}, b_2 + \maxo{b_1 - g_2}),
\end{align*}
where $\charac{x}{P_i} = (g_i, b_i)$\@.
It is not hard to show that the characteristic $\charac{x}{P}$ is well-defined.

As shown in Proposition~\ref{prop:4} below, the concept of characteristic is closely related to that of the \textit{differential.}
Where $g_x$ and $b_x$ are the number of global and buffer write actions to $x$ in a pomset $P$, respectively, call $\differ{x}{P} = b_x - g_x$ the \defin{differential} of $x$ in $P$\@.

\begin{prop}%
  \label{prop:4}
  For all $P \in \poms(\Atso)$, if $\charac{x}{P} = (g, b)$, then $\differ{x}{P} = b - g$\@.
  \qed%
\end{prop}

\noindent The following proposition tells us that the number of buffer entries for $x$ in the final state of a footstep of a pomset $P$ is determined by the differential of $x$ in $P$ and by the number of buffer entries for $x$ in the initial state.
For compactness, let $\sigma[\bar x] = n$ whenever $\sigma(\bar x) = v_n$\@.
\begin{prop}%
  \label{prop:3}
  For all $P \in \poms(\Atso)$, $\Lambda \in \Gamma(P)$, $(\sigma, \tau) \in \footp{P}_\Lambda$, and $\bar x \in \dom(\sigma)$, if $\bar x \in \dom(\tau)$, then $\tau[\bar x] = \sigma[\bar x] + \differ{x}{P}$\@.
  \qed%
\end{prop}

\begin{cor}%
  \label{cor:2}
  Let $P \in \poms(\Atso)$, $\Lambda \in \Gamma(P)$, $(\sigma, \tau) \in \footp{P}_\Lambda$, and $\charac{x}{P} = (g, b)$\@.
  Then $\tau[\bar x] = \sigma[\bar x] + b - g$\@.
  In particular, when $g = 0$, then $\tau[\bar x] = \sigma[\bar x] + b$\@.
  \qed%
\end{cor}

\noindent The following lemma is useful in bounding the characteristic of $P$\@.
In particular, it implies that if there exists a $\Lambda \in \Gamma(P)$ and a $(\sigma, \tau) \in \footp{P}_\Lambda$ such that $\zeta(\sigma)$ and $\zeta(\tau)$, then $\charac{x}{P} = (0, 0)$\@.

\begin{lem}%
  \label{lemma:10}
  Let $P \in \poms(\Atso)$, $\Lambda \in \Gamma(P)$, $(\sigma, \tau) \in \footp{P}_\Lambda$, and $\charac{x}{P} = (g, b)$\@.
  Then $\sigma[\bar x] \geq g$, and if $\bar x \notin \dom(\sigma)$, then $g = 0$\@.
  Moreover, $\tau[\bar x] \geq b$, and if $\bar x \notin \dom(\tau)$, then $b = 0$\@.
  \qed%
\end{lem}

\begin{cor}%
  \label{cor:3}
  If $(\sigma, \tau) \in \footp{P_1 \parop P_2}_\Lambda$, then $\charac{x}{P_i} = (0, 0)$ and $\charac{x}{P} = (0, 0)$\@.
  \qed%
\end{cor}

\noindent In certain cases, we can ``read off'' from $P$ and $\Lambda$ what certain values in a footstep should be.
We can show by induction on the derivation of $(\sigma, \tau) \in \footp{P}_\Lambda$ that:

\begin{prop}\label{prop:5}
  Let $P \in \poms(\Atso)$, $\Lambda \in \Gamma(P)$, $(\sigma, \tau) \in \footp{P}_\Lambda$ and let $\charac{x}{P} = (g, b)$\@.
  Let ${\{\rho^{(i)}\}}_{i=1..k}$ be the minimal reads ${\{x = v^{(i)}\}}_{i=1..k}$ from $x$ under $<_P$\@.
  \begin{enumerate}[label={(\arabic*)}, ref={(\arabic*)}]
  \item If there exist $i_1, \dotsc, i_l$, $l \geq 1$, such that there are no writes to $x$ or $\bar x$ in $\bigcup_{j=1}^l \lc{\rho^{(i_j)}}\Lambda$, then $\bar x \in \dom(\sigma)$, there exists a $v$ such that $v = v^{(i_1)} = \dotsb = v^{(i_l)}$, $\sigma(\bar x) = v_n$ for some $n \geq g$, and $\sigma(x) = v$ when $n = 0$\@.
  \item If there exists a maximal global write $x := v$ to $x$ under $<_\Lambda$ (it may be in $P$ or not), then $\tau(x) = v$\@.
    Moreover, $x \in \dom(\tau)$ if and only if there exists a global write to $x$ in $\Lambda$\@.
  \item If there exists a maximal buffer write $\bar x := v$ to $x$ under $<_\Lambda$, then $\tau(\bar x) = v_n$ for some $n \geq b$\@.
    Moreover, $\bar x \in \dom(\tau)$ if and only if there exists a global write to $x$ or a buffer write to $\bar x$ in $P$\@.
  \item If there exist no buffer writes $\bar x := v$ and $\bar x \in \dom(\tau)$, then where $\sigma(\bar x) = v_n$, we have $\tau(\bar x) = v_m$ for some $m$\@.
  \end{enumerate}
\end{prop}

\noindent As a corollary, we validate the intuition that the final state should be determined by the total order imposed by $\Lambda$ on the writes, where $\beta$ is as in~\eqref{eq:6}:
\begin{cor}%
  \label{cor:l-determines-tau}
  For all $P$, $\Lambda \in \Gamma(P)$, and $(\sigma, \tau) \in \footp{P}_\Lambda$, we have $\tau \restriction_\Loc = \lsState{\Lambda \restriction_\Aw}$\@.
  \qed%
\end{cor}

\subsubsection{Executions}

Let $\zfootp{P}_\Lambda = \{ (\sigma, \tau) \in \footp{P}_\Lambda \mid \zeta(\sigma) \aand \zeta(\tau)\}$ be the subset of footsteps with empty buffers.
The set of \defin{TSO executions} of a finite pomset $P$ is given by the set $\DExec P = \{ (\sigma, \supd{\sigma}{\tau}) \mid \Lambda \in \Lin(P), (\sigma', \tau) \in \zfootp{P}_\Lambda, \sigma' \subseteq \sigma\}$; infinite pomsets are considered in Section~\ref{sec:infinite-executions}.
These executions take all of the states $\sigma$ containing a minimal fragment $\sigma'$ required to execute $P$ to that state updated with the effects $\tau$ of $P$\@.
The set of TSO executions for a program $p$ is then $\DExec p = \bigcup_{P \in \PTSO(p)} \DExec P$\@.
We say that a finite pomset $P$ and a program $p$ are \defin{TSO executable} if $\DExec P$ and $\DExec p$ are non-empty, respectively.

We illustrate TSO pomset executions by validating the IRIW litmus test, \ie, by showing that all writes appear in the same order to all threads.
For example, starting from a state initialized to zero, executing the program \[
  {x:=1} \parop {y:=1} \parop {(w_1:=x; w_2:=y)} \parop {(z_1:=y; z_2:=x)}
\]
under TSO should never give a state consistent with $\state{w_1:1,w_2:0,z_1:1,z_2:0}$\@.
To show this, it is sufficient to show that the following pomset $P$ is not executable: \[
  \xymatrix@C-1em@R-1.5em{\bar x:=1 \ar[d]&\bar y:=1\ar[d]& x = 1 \ar[d]
    & y = 1 \ar[d]\\
    x:=1 & y:=1 & y = 0 & x = 0.}
\]
Consider some $\Lambda \in \Lin(P)$\@.
Without loss of generality, assume $x:=1 <_\Lambda y:=1$\@.
To get an execution, we must apply \rn{Par}, and eventually we will need to compute $\footp{P_4}_{\Lambda_4}$ where $P_4$ is $y = 1 \to x = 0$ and $\Lambda_4$ is such that $x:=1 <_{\Lambda_4} y:= 1$\@.
To be able to execute $y=1$ and still get a footstep with an empty initial buffer, we need $\Lambda_4$ to satisfy $y:=1 < y=1$\@.
But then $\footp{P_4}_{\Lambda_4} = \footp{x:=1, y:= 1}^* \seqop \footp{\{y=1\}}_{\lst{y=1}} \seqop \footp{\{x~=~0\}}_{\lst{x=0}}$, and there are no footsteps in $\footp{x:=1, y:= 1}^* \seqop \footp{\{y=1\}}_{\lst{y=1}}$ that can be combined with those in $\footp{\{x=0\}}_{\lst{x=0}}$ to get states with empty buffers.
This means we cannot combine footsteps from $P_4$ to get footsteps for $P$ using the \rn{Par} rule, and so $\footp{P}_{\Lambda}$ will be empty.

As discussed in the introduction, the Dekker mutual exclusion algorithm fails under TSO\@.
Indeed, the second pomset for Dekker on page~\pageref{dekker-tso-poms} can be executed from an initial state having both $x$ and $y$ set to zero.
To do so, we take a $\Lambda$ such that ${y~=~0} <_\Lambda {y~:=~1}$ and ${x~=~0} <_\Lambda {x~:=~1}$, and apply \rn{Par} followed by \rn{Seq}.

In contrast, the Peterson algorithm successfully enforces mutual exclusion under TSO\@.
Consider the following instance of the Peterson algorithm:
\[ {(x:=1; \ifb{x=2}{l:=1}{\cskip})} \parop {(x:=2;\ifb{x=1}{r:=1}{\cskip})}. \] Starting from the initial state $\state{x:0,l:0,r:0}$, one cannot execute the above under TSO and reach a state where both $l$ and $r$ are 1.
In showing this, we can safely ignore all pomsets where a read from $x$ appears before the global write to $x$, because whenever we have $\bar x:=v \to x=v' \to x:=v$ in a command's TSO pomset, we must have $v = v'$\@.
This implies that if a thread reads $x$ before it does the global write to $x$, it will take the $\cskip$ branch of the conditional.
It is then sufficient to show that the following pomset is not executable: \[
  \xymatrix@C-1em@R-1.5em{{\bar x} :=1 \ar[r] & x:= 1 \ar[r] & x = 2\\
    {\bar x} := 2 \ar[r] & x:= 2 \ar[r] & x=1.}
\]
Consider some $\Lambda \in \Lin(P)$\@.
Without loss of generality, assume $x:=1 <_\Lambda x:=2$\@.
To get an execution, we must apply \rn{Par} and derive a footstep for the bottom row $P_2$ under some $\Lambda_2$ where $x:=1 <_{\Lambda_2} x:=2$\@.
To compute this footstep, we must repeatedly apply \rn{Seq}, and will eventually reach the stage where $\footp{P_2}_{\Lambda_2} = \{(\state{x~:~0},\state{x~:~2})\}\seqop\footp{x=1}_{\lst{x=1}}$\@.
But this footprint must be empty, because $\supd{x:0}{x:2}$ is not consistent with $\state{x:1}$\@.
We thus cannot apply \rn{Par} and we conclude that the pomset is not executable.
It follows that the Peterson algorithm enforces mutual exclusion under TSO\@.

\subsection{Fences}

We can extend the above semantics to deal with fences.
This extension will not be referenced in subsequent sections, and we mention it here merely to emphasize the flexibility of our general development.

A \textit{fence} constrains the reordering of memory actions.
To capture fences, we first introduce a command $\cfence$\@.
Under TSO, fences cause all actions before the fence to be observed before any actions after the fence.
It is sufficient to flush the thread's buffer to ensure this, giving rise to the semantic clause $\PT_L(\cfence) = \{(L;\{\delta\},\emptylist)\}$\@.

\subsection{Infinite executions}%
\label{sec:infinite-executions}

Though we cannot capture the input-output behaviour of infinite pomsets (such executions have no ``final state''), we can describe their executions in terms of executions of finite pomsets.
We first characterize the ``phased executions'' describing infinite executions.
Then we describe when a phased execution is an execution of a given infinite pomset.

\newcommand{\PE}{\calP_\mathcal{E}}
The set $\PE(\sigma) \subseteq \poms(\Sigma \times \PTSO \times \Sigma)$ of \defin{execution-from-$\sigma$ pomsets} is coinductively defined by:
\begin{enumerate}
\item if $P$ is finite and $(\sigma, \tau) \in \DExec P$, then $\{ (\sigma, P, \tau)\} \in \PE(\sigma)$; and
\item if $P$ is finite, $(\sigma, \tau) \in \DExec P$, and $E \in \PE(\supd{\sigma}{\tau})$, then $\{(\sigma,P,\tau)\} ; E \in \PE(\sigma)$\@.
\end{enumerate}
All pomsets in $\PE(\sigma)$ are linear, satisfy the finite-height property, and are series-parallel.
Given an $E \in \PE(\sigma)$, its underlying TSO pomset $\pi_2(E)$ is obtained by sequentially composing its pomsets in the obvious manner.
Explicitly, $\pi_2(E) = (\bigcup_{e \in E} \{e\} \times (\pi_2\circ \Phi_E)(e), {<}, \Phi)$, where $\Phi(e, p) = \Phi_{\Phi_E(e)}(p)$, and $(e_0,p_0) < (e_1,p_1)$ if and only if $e_0 <_E e_1$, or $e_0 = e_1$ and $p_0 <_{\Phi(e_0)} p_1$\@.

An execution pomset $E$ is a \defin{phased execution} of $P \in \PTSO$ if $P = \pi_2(E)$\@.
Let $\PE(P) = \bigcup_{\sigma \in \Sigma} \{ E \in \PE(\sigma) \mid \pi_2(E) = P \}$ be the set of execution pomsets of $P$\@.

This definition captures the ``termination'' axiom~\cite[p.~283]{SPARC8} that guarantees fair executions.
A pomset $P$ satisfies this axiom if for all write actions $w$ to $x$ and all infinite sequences $r_1 <_P r_2 <_P r_3 <_P \cdots$ of read actions from $x$, there exists a $j$ such that $w <_P r_j$\@.
This axiom prohibits, \eg, the non-terminating execution of the program $y := 1 \parop \whileb{y=0}{\cskip}$\@.
For every $E \in \PE(\sigma)$, $\pi_2(E)$ satisfies the termination axiom.
Indeed, consider a write action $w$ and an infinite sequence ${\{r_i\}}_{i \in \N}$ in $\pi_2(E)$\@.
There exists an $e \in E$ such that $w \in e$\@.
Because each of the pomsets in $E$ is finite and $E$ is linear, there exists an $e'$ and an $i$ such that $e <_E e'$ and $r_i \in e'$\@.
Then $w <_{\pi_2(E)} r_i$ by definition.

\begin{prop}
  We can characterize phased executions of pomsets $P$ as follows:
  \begin{enumerate}
  \item if $P$ is finite, then every $E$ in $\PE(P)$ is finite;
  \item if some $E \in \PE(P)$ is finite, then $P$ is finite;
  \item if $P$ is finite, $(\sigma, \tau) \in \DExec P$, then $(\sigma, P, \tau) \in \PE(P)$;
  \item if $\{(\sigma_1,P_1,\tau_2)\};\cdots;\{(\sigma_n,P_n,\tau_n)\} \in \PE(P)$ is finite, then $P = P_1;\cdots;P_n$, $\sigma_{i+1} = \supd{\sigma_i}{\tau_i}$ for $1 \leq i < n$, and $(\sigma_1, \supdm{\sigma_1 \mid \tau_1 \mid \cdots \mid \tau_n}) \in \DExec P$\@. 
  \end{enumerate}
\end{prop}

\section{Soundness and Completeness}%
\label{sec:sands}

We show that our denotational account of TSO in Section~\ref{sec:denotational} is sound and complete relative to the axiomatic account of Section~\ref{sec:axiom}.
Soundness implies that we capture only behaviours permitted by the axiomatic account; completeness implies that we capture all behaviours permitted by the axiomatic account.
Because all TSO-consistent orders are contained in TSO-consistent total orders by Corollary~\ref{cor:tsotot} and can be obtained by weakening these, it is sufficient to show that we capture all TSO-consistent total orders.
We identify total orders and lists.

\subsection{Soundness}

We call a function $f : \poms(\Apo) \to \wp(\lists{\Apo})$ \defin{sound} when for every program $p$ and finite pomset $P \in \Ppo(p)$, if $L \in f(P)$, then $L$ is TSO-consistent with $P$\@.

We will construct such an $f$ and show that it is sound in this subsection, and we will show that it is complete in the next subsection.
We begin by giving an important characterization of pomsets for programs.

\begin{lem}%
  \label{lemma:1}
  For every program $p$ and pomset $P \in \PTSO(p)$, there exists an order isomorphism $\omega : P\restriction_\Ab \to P\restriction_\Aw$ such that if $\omega(\bar x := v) = (y := w)$, then $y = x$ and $w = v$, and such that for all $b \in P\restriction_\Ab$, we have $b <_P \omega(b)$\@.
  \qed%
\end{lem}

For programs, this means that global writes appear after the corresponding write to the buffer, that global writes occur in the same order as the writes to the buffer, and that all writes to the buffer give rise to global writes.
We call a series-parallel TSO pomset for which such an $\omega$ exists for each of its SP components \textit{well-balanced}; these ``patch together'' to form an $\omega$ for the entire pomset.
Given a program $p$, all TSO pomsets in $\PTSO(p)$ are well-balanced, and finite $P \in \PTSO(p)$ have a unique such $\omega$\@.
We assume in the rest of this section that our TSO pomsets are well-balanced.

Consider the function $U : \poms(\Atso) \to \poms(\Apo)$ that takes each TSO pomset to its underlying program order.
It does so by first deleting all global write actions, and then relabelling all buffer write actions $\bar x:=v$ by corresponding global write actions $x:=v$\@.
We identify all reads in $P$ with the corresponding reads in $U(P)$ and $(x:=v) = \omega(\bar x := v)$ with the write $x:=v$ below $\bar x:=v$\@.
We can imagine $U$ and the identifications as being given in the following diagram, where dashed arrows indicate identifications, solid arrows indicate the pomset orders, and the $\lambda_j$ are arbitrary read actions: \[
  \xymatrix@C-1em{
    \cdots \ar[r] &
    \lambda_i \ar[r] \ar@{-->}[d] & \bar x := v \ar[r] & \lambda_{i+2} \ar[r]\ar@{-->}[d] &
    \cdots \ar[r] &
    \lambda_k \ar[r]\ar@{-->}[d] & x := v \ar[r]\ar@{-->} `d[l]
    `l[lllld] [llldl]  & \lambda_{k+2} \ar[r]\ar@{-->}[d] &
    \cdots & &
    P \ar[d]^U\\ 
    \cdots \ar[r] &
    \lambda_i \ar[r] & x := v \ar[r] & \lambda_{i+2} \ar[r] &
    \cdots \ar[r] &
    \lambda_k \ar[rr] & & \lambda_{k+2} \ar[r] &
    \cdots & &
    U(P)
  }
\]

By observing that the PO pomset clauses are essentially special cases of the TSO pomset clauses, we have that for all programs $p$, $U(\PTSO(p)) \subseteq \Ppo(p)$\@.
This inclusion is actually an equality, because given any program order $P$ for $p$, we can construct a TSO pomset $P'$ such that $U(P') = P$ by immediately flushing the buffer with $\splt$ after every write, \ie, by replacing all occurrences of $x:=v$ in $P$ with $\bar x := v \to x:=v$ to get $P'$\@.

Let the set $\T(P)$ of TSO-consistent total orders of $P \in \poms(\Atso)$ be given by \[
  \T(P) = \bigcup_{P' \in U^{-1}(P)} \{ \Lambda\restriction_\Apo \mid \Lambda \in \Lin(P') \aand \zfootp{P'}_\Lambda \neq \emptyset \}.
\] Informally, $\T(P)$ captures the linearisations of TSO pomsets in $U^{-1}(P)$ that give rise to TSO executions of pomsets.

We will eventually show that $\T$ is sound.
First, we prove a few technical lemmas that will be useful in showing axiom \rn{V} is satisfied.
The first implies that \rn{Va} is satisfied.

\begin{lem}%
  \label{lemma:2}
  Let $P$ be a TSO pomset and $\Lambda \in \Gamma(P)$ such that $\zfootp{P}_\Lambda \neq \emptyset$\@.
  Let ${(x=v)}_r \in P$ be a read action such that for all buffered write actions $b <_P {(x=v)}_r$, its corresponding global write action $\omega(b)$ satisfies $\omega(b) <_P {(x=v)}_r$\@.
  If there exists a write ${(x:=v')}_w$ maximal under $<_\Lambda$ amongst all writes to $x$ in $\lc{{(x=v)}_r}\Lambda$ and all global writes to $x$ in $\lc{{(x=v)}_r}P$ are in $\lc{{(x:=v')}_w}\Lambda$, then $v' = v$\@.
\end{lem}

\begin{proof}
By the remarks on page~\pageref{pars:sp-explan}, it is sufficient to consider only linear pomsets $P$\@.
Let $\Lambda_L = \lc{r}\Lambda \setminus r$, $\Lambda_R$ be such that $\Lambda = \Lambda_L;\Lambda_R$, $L = P \cap \Lambda_L$, and $R = P \cap \Lambda_R$\@.
Let $(g, b) = \charac{x}{P}$, then by Lemma~\ref{lemma:10}, we have $g = b = 0$\@.
Consider some arbitrary footstep $(\sigma, \tau) \in \zfootp{P}_\Lambda$, and let $(\sigma_L, \tau_L) \in \footp{L}_{\Lambda_L}$ and $(\sigma_R, \tau_R) \in \footp{R}_{\Lambda_R}$ be the footsteps combined by \rn{Seq} to form $(\sigma, \tau)$\@.
Then because $\sigma_L \subseteq \sigma$ and $\zeta(\sigma)$, we have $\zeta(\sigma_L)$\@.
Because $L$ has the same number of buffer writes to $\bar x$ as global writes to $x$, we have $\charac{x}{L} = (0, 0)$ by Lemma~\ref{lemma:10} and Proposition~\ref{prop:4}.
So by Corollary~\ref{cor:2}, we get $\tau_L[\bar x] = \sigma[\bar x] = 0$\@.
By Proposition~\ref{prop:12}, we then get $\tau_L(x) = v'$, and also that $\sigma_R(x) = v$\@.
Because $\scons{\supd{\sigma_L}{\tau_L}}{\sigma_R}$ by \rn{Seq}, we conclude $v' = v$\@.
\end{proof}

\noindent The following lemma implies that \rn{Vb} is satisfied.

\begin{lem}%
  \label{lemma:4}
  Let $P$ be a TSO pomset and $\Lambda \in \Gamma(P)$ such that $\zfootp{P}_\Lambda \neq \emptyset$\@.
  Let ${(x=v)}_r \in P$ be a read action such that there exists a buffered write action $b' <_P {(x=v)}_r$ writing to $x$\@.
  Let ${(\bar x:=v')}_b$ be the maximal such $b'$ with regards to $<_P$\@.
  If ${(x=v)}_r <_P \omega({(\bar x:=v')}_b)$, then $v = v'$\@.
\end{lem}

\begin{proof}
It is again sufficient to consider only linear $P$\@.
Let $\Lambda_L$, $\Lambda_R$, $L$, $R$, $(\sigma_L, \tau_L)$, and $(\sigma_R, \tau_R)$ be as in the proof of Lemma~\ref{lemma:2}.
By Proposition~\ref{prop:3}, we have $\tau_L[\bar x] = \sigma_L[\bar x] + b_x - g_x$, where $b_x$ and $g_x$ are the number of buffer writes and global writes to $x$ in $L$, respectively.
But $\sigma_L[\bar x] = \sigma[\bar x] = 0$, and well-balancedness and the last hypothesis imply $b_x > g_x$, so $\tau_L[\bar x] > 0$ and $\tau_L(\bar x) = v'_m$ for some $v$ and $m > 0$\@.
But $\sigma_R(\bar x) = v_n$ for some $v$ and $n$ as well.
Because $\scons{\supd{\sigma_L}{\tau_L}}{\sigma_R}$ by \rn{Seq}, we conclude $v' = v$\@.
\end{proof}

\begin{thm}
  The function $\T$ is sound.
\end{thm}

\begin{proof}
Let $p$ be an arbitrary program, and let $P' \in \Ppo(p)$, $P \in U^{-1}(P')$, and $\Lambda \in \Lin(P)$ be arbitrary such that there exists a $(\sigma, \tau) \in \zfootp{P}_\Lambda$\@.
Let $L = \Lambda\restriction_\Apo$\@.
We check the six axioms to show that $L$ is TSO-consistent for $U(P) = P'$ from $\sigma$\@.

Axiom \rn{O}.
This is clearly satisfied, because $U$ leaves reads untouched, and there is an order isomorphism between global writes in $P$ and global writes in $U(P)$ by Lemma~\ref{lemma:1}, and linearisations preserve all elements.

Axiom \rn{Va}.
Let ${(x = v)}_r \in U(P)$ be an arbitrary read, and assume we have some write ${(x:=v')}_w$ maximal amongst all writes to $x$ in $\lc{{(x=v)}_r}{L}$, and that all writes to $x$ in $\lc{{(x=v)}_r}{U(P')}$ are in $\lc{{(x:=v')}_w}{L}$\@.
That's to say, assume the hypotheses to \rn{Va} hold for ${(x=v)}_r$\@.
We must show $v = v'$\@.
Because all writes to $x$ in $\lc{{(x=v)}_r}{U(P')}$ are in $\lc{{(x:=v')}_w}{L}$, this means all buffer writes to $x$ in $\lc{{(x=v)}_r}{\Lambda}$ are in $\lc{{(x:=v')}_w}{L}$\@.
Because $w$ appears in $L$, it is a global write, and so all of the corresponding global writes for the buffer writes to $x$ in $\lc{{(x=v)}_r}{\Lambda}$ are also in $\lc{{(x:=v')}_w}{L}$\@.
Then we can apply Lemma~\ref{lemma:2} and conclude $v = v'$\@.

Axiom \rn{Vb}.
Let ${(x = v)}_r \in U(P)$ be an arbitrary read.
Assume there exists a write ${(x:=v')}_w$ to $x$ maximal under $<_P$ amongst all writes to $x$ in $\lc{{(x=v)}_r}{U(P)}$, but that ${(x=v)}_r <_L {(x:=v')}_w$\@.
We must show that $v' = v$\@.
That ${(x:=v')}_w <_{U(P)} {(x=v)}_r$ implies ${(\bar x :=v')}_w <_P {(x=v)}_r$, and ${(x=v)}_r <_L {(x:=v')}_w$ implies ${(x=v)}_r <_P \omega({(\bar x :=v')}_w)$\@.
So by Lemma~\ref{lemma:4}, we have $v' = v$\@.

Axiom \rn{Vc}.
Let ${(x = v)}_r \in U(P)$ be an arbitrary read.
If there exist no writes ${(x := v')}_w <_{U(P)} {(x=v)}_r$ in $U(P)$, then there exist no buffered writes ${(\bar x := v')}_w <_P {(x=v)}_r$ in $P$, and so by Proposition~\ref{prop:5}, axiom \rn{Vc} is satisfied.

Axiom \rn{L}.
Consider some arbitrary read action $r \in U(P)$ and action $a$ such that $r <_{U(P)} a$\@.
Because linearisation preserves order, it is sufficient to show that $r <_P a$\@.
If $a$ is also a read action, then $U$ leaves $a$ untouched, and so $r <_p a$\@.
If $a$ is a write action, then there exists a corresponding buffered write $a'$ in $P$ such that $r <_P a'$ and $\omega(a') = a$\@.
By Lemma~\ref{lemma:1}, $a' < \omega(a')$, so $r <_P a$ by transitivity.

Axiom \rn{S}.
Lemma~\ref{lemma:1} implies there exists an order isomorphism between $P\restriction\Aw$ and $U(P)\restriction\Aw$, so $w <_P w'$ implies $w <_{U(P)} w'$\@.
Linearisation preserves order, so $w <_{U(P)} w'$ implies $w <_\Lambda w'$\@.
Because $w$ and $w'$ are global writes, $w <_L w'$ as desired.

Axiom \rn{F}.
Let $\alpha_1, \alpha_2, \alpha_3 \in U(P)$ be arbitrary such that $\alpha_1 <_{U(P)} \alpha_3$, $\alpha_1 <_{U(P)} \alpha_3$, and $\alpha_2 \nc_{U(P)} \alpha_3$\@.
If $\alpha_1$ is a write, then it is sufficient to observe that by the definition of TSO pomsets for parallel compositions of expressions or commands, for any buffered write $b$ such that $b <_P \alpha_2$ and $b <_P \alpha_3$ for some $\alpha_2 \nc_P \alpha_3$, we have $\omega(b) <_P \alpha_2$ and $\omega(b) <_P \alpha_3$\@.
For then $\alpha_1 <_L \alpha_2$ and $\alpha_1 < \alpha_3$\@.
If $\alpha_1$ is a read, then we are done by \rn{L}.

Axiom \rn{J}.
Let $\alpha_1, \alpha_2, \alpha_3 \in U(P)$ be arbitrary such that $\alpha_1 <_{U(P)} \alpha_3$, $\alpha_2 <_{U(P)} \alpha_3$, and $\alpha_1 \nc_{U(P)} \alpha_2$\@.
If $\alpha_1$ and $\alpha_2$ are reads, then we are done by \rn{L}, so assume without loss of generality that $\alpha_1$ is some write ${(x:=v)}_w$\@.
Then $\alpha_1$ corresponds to some buffered write ${(\bar x:=v)}_w$ in $P$\@.
By the dual observation to that in \rn{F}, we have that if $b$ is a buffered write such that for some $a$ and $c$, $a <_P c$ and $b <_P a$ but $a \nc_P b$, then $\omega(b) <_P a$\@.
So we get $\alpha_1 <_L \alpha_3$ as desired.
A symmetric argument applies when $\alpha_2$ is a write.
\end{proof}

\subsection{Completeness}%
\label{sec:completeness}

We call a function $f : \poms(\Apo) \to \wp(\lists{\Apo})$ \defin{complete} when for every program $p$ and finite pomset $P \in \Ppo(p)$, if $L\in \lists{\Apo}$ is TSO-consistent with $P$, then $L \in f(P)$\@.

Our goal is to show that $\T$ is complete.
To do so, we first construct a pomset $s(P,L) \in U^{-1}(P)$ for any $P \in \poms(\Apo)$ and total order $L$ that is TSO-consistent with $P$\@.
Let the underlying set $s(P,L)$ be given by $\{0\}\times P \cup \{1\}\times P\restriction_\Aw $\@.
Let $\Phi_{s(P,L)}$ be given by $\Phi_{s(P,L)}((i,p)) = (\bar x := v)$ if both $i = 0$ and $\Phi_P(p) = (x:=v)$, and $\Phi_{s(P,L)}((i,p)) = \Phi_P(p)$ otherwise.
Intuitively, $<_{s(P,L)}$ merges the global writes into the program order in the places specified by $L$\@.
Let $<_{s(P,L)}$ be the least strict partial order generated by the following collection of inequalities:
\begin{enumerate*}[label={\textit{(\roman*)}},ref={\textit{(\roman*)}}]
\item $(0,p) <_{s(P,L)} (0,p')$ if $p <_P p'$;
\item $(1,w) <_{s(P,L)} (1,w')$ if $w \comp_P w'$ and $w <_L w'$;
\item $(0,w) <_{s(P,L)} (1, w)$ if $w \in P\restriction_\Aw$;
\item $(0,p) <_{s(P,L)} (1, w)$ if $p \comp_P w$ and $p <_L w$; and
\item $(1,w) <_{s(P,L)} (0, p)$ if $w \comp_P p$ and $w <_L p$\@.
\end{enumerate*}
We identify the program order actions in $P$ with their corresponding instances in $s(P,L)$, that is, we identify all read actions $r$ of $P$ with $(0,r)$ in $s(P,L)$, and all global write actions $w$ of $P$ with the global write action $(1,w)$ in $s(P,L)$\@.
We leverage this identification below to lift $L$ from $P$ to $s(P,L)$\@.

The proof of completeness can be broken down into three lemmas as follows:

\begin{lem}%
  \label{lemma:6}
  If $p$ is a program, $P \in \Ppo(p)$, and $L$ is TSO-consistent with $P$, then $s(P,L) \in \PTSO(p)$ and $U(s(P,L)) = P$\@.
  \qed%
\end{lem}

\begin{lem}%
  \label{lemma:5}
  If $P \in \poms(\Apo)$ and $L$ is TSO-consistent with $P$, there exists a $\Lambda \in \Lin(s(P,L))$ such that $\Lambda\restriction_\Apo = L$\@.
  \qed%
\end{lem}

\begin{lem}%
  \label{lemma:7}
  If $p$ is a program, $P \in \PTSO(p)$, $L_0 \in \Ls$, $L$ is TSO-consistent with $U(P) \parop L_0$, $\Lambda \in \Lin(P \parop L_0)$, and $\Lambda\restriction_\Apo = L$, then $\zfootp{P}_\Lambda \neq \emptyset$\@.
  \qed%
\end{lem}

\noindent The presence of $L_0$ in the statement of Lemma~\ref{lemma:7} lets us use an induction hypothesis in the case where $P$ is a parallel composition of pomsets, because when we want to apply the induction hypothesis to one of the pomsets, we need to be able to reference the global writes in the other pomset.

\begin{thm}
  The function $\T$ is complete.
\end{thm}

\begin{proof}
Given a program $p$, a $P \in \Ppo(P)$ and an $L$ that is TSO-consistent with $P$, we have by Lemma~\ref{lemma:6} that $s(P,L) \in U^{-1}(P)$\@.
By Lemma~\ref{lemma:5}, we have a $\Lambda \in \Lin(s(P,L))$ such that $\Lambda\restriction_\Apo = L$\@.
By Lemma~\ref{lemma:7} with $L_0 = \emptylist$, $\zfootp{s(P,L)}_\Lambda \neq \emptyset$\@.
So $L \in \T(P)$, and we conclude completeness.
\end{proof}

\section{Related Work}%
\label{sec:related-work}

Other approaches to semantics for weak memory models mostly use execution graphs and operational semantics.
Execution graphs~\cite{Batty2011,Boehm2008} serve to describe the executional behaviour of an entire program, an inherently non-modular approach.
We see our denotational framework as offering an alternative basis for program analysis, compositional and modular by design.
Boudol and~Petri~\cite{Boudol2009} gave an operational semantics framework for weak memory models that uses buffered states.
Jagadeesan et~al.\@~\cite{Jagadeesan2012} adapted a fully abstract, trace-based semantics by Brookes~\cite{Brookes1996} to give a fully abstract denotational semantics for TSO\@.
Higham and~Kawash~\cite{Higham:2000:MCP:645446.653212} gave an axiomatic semantics for TSO using linearizations.
They showed their semantics to be equivalent to an abstract machine with buffers based on the informal description in~\cite{SPARC8}.
Their account does not consider forking and joining of processes.
Demange et~al.\@~\cite{Demange:2013:PBB:2429069.2429110} proposed a tractable memory model consistent with the Java Memory Model.
They give an operational semantics using buffers and an axiomatic account and show the two to be equivalent.
Sewell et~al.\@~\cite{Sewell:2010:XRU:1785414.1785443} gave a TSO-like memory model for x86, formalized in HOL4.
Alglave et~al.\@~\cite{Alglave:2014:HCM:2633904.2627752} gave a generic axiomatic framework for describing weak memory models and showed how to specialize it to various memory models.
Jeffrey and~Riely~\cite{Jeffrey2019} give a denotational account of relaxed memory models using event structures.

Pratt~\cite{Pratt1986} was the first to generalize from traces to pomsets in the study of concurrency.
He introduced the parallel composition operations we presented in Section~\ref{sec:denotational} and he used pomsets in the study of concurrent processes.
Building on Pratt's work, Brookes~\cite{Brookes2015,Brookes2016MFPS,Brookes2016} introduced a pomset framework to study weak memory.
This framework used Pratt's parallel composition operator, and its sequential composition is a variant of Pratt's concatenation operation.
The TSO semantics given above builds on Brookes's work.
The key technical differences involve adapting pomset semantics to incorporate state equipped with abstract buffers, with careful accounting to deal properly with order relaxations allowed by TSO\@.
Our formal axiomatization of SPARC TSO, including full treatment of forks and joins, is a crucial part of the set-up, allowing us to be precise about the relationship between our abstract denotational semantics and the more concrete and informal characterization of TSO that appears in the manual.

\section{Conclusion}

Our denotational semantics accurately captures the behaviours of SPARC TSO, and its compositionality enables us reason modularly about programs.
The main strength of the pomset approach is its conceptual simplicity.
Unlike trace semantics, which include irrelevant orderings of actions, our pomset semantics specifies only relevant orderings of actions.
Its simplicity also makes it readily adaptable to other memory models.
For example, to capture SPARC PSO, which relaxes TSO to provide only per-location global orders on writes, we conjecture that it is sufficient to replace the single buffer parameter $L$ in the $\PB_L$ semantic clauses with families ${\{L_x\}}_{x \in \Loc}$ of buffers, and to then modify the buffer flushing clauses in the obvious manner and to handle store barrier ``\texttt{stbar}'' commands.
Memory models provided by modern processors are often much weaker than TSO and PSO\@.
To model these, we conjecture that it is sufficient to specify the correct set of actions, and to modify the semantic clauses generating pomsets and footsteps accordingly.
We believe pomset semantics provide fertile ground for future research in semantics for weak memory models.

\section*{Acknowledgements}

This paper is an extended version of one~\cite{KAVANAGH2018223} presented at MFPS~XXXIII\@.
The authors gratefully acknowledge feedback from anonymous reviewers and from André Platzer.

\bibliographystyle{alpha}
\bibliography{lmcs}

\end{document}